\documentclass[12pt]{article}
\usepackage[total={6.5in,8.75in},
top=1in, left=0.9in, includefoot]{geometry} 

\usepackage{dsfont}
\usepackage{psfrag}
\usepackage{calrsfs}
\usepackage{mathrsfs}
\usepackage{amsmath, amsthm, amssymb, amsfonts}
\usepackage{mathtools}
\usepackage[shortlabels]{enumitem}
\usepackage[comma]{natbib}
\setcitestyle{round}
\usepackage{url}
\usepackage{float}
\usepackage{multirow}
\usepackage{acronym}
\usepackage[usenames]{color}
\usepackage{tikz}
\usetikzlibrary{arrows,decorations.pathmorphing,backgrounds,positioning,fit,automata}
\usetikzlibrary{petri}
\usetikzlibrary{topaths}
\usetikzlibrary{shapes,snakes}
\usepackage{indentfirst,calc,euscript}
\usepackage{setspace}
\usepackage[reals]{layout}
\usepackage{xr}
\usepackage{amscd}
\usepackage{sgame}
\usepackage{graphicx}
\usepackage{subcaption}
\usepackage[textwidth=30mm]{todonotes}
\usepackage[colorlinks=true,breaklinks=true,bookmarks=true,urlcolor=magenta,
     citecolor=blue,linkcolor=red,bookmarksopen=false,draft=false]{hyperref}
\usepackage{colortbl}
\newcommand{\card}{\operatorname{card}}

\newcommand{\degree}{\operatorname{degree}}

\newcommand{\IPoA}{\operatorname{\mathsf{PoA}}}
\newcommand{\IPoS}{\operatorname{\mathsf{PoS}}}

\newcommand{\PoA}{\operatorname{\mathsf{PoA}}}
\newcommand{\PoS}{\operatorname{\mathsf{PoS}}}

\newcommand{\diff}{\ \mathrm{d}}

\newcommand{\OUT}{\textup{OUT}}
\newcommand{\simplex}{\mathcal{S}}

\newcommand{\ignore}[1]{}
\newtheorem{theorem}{Theorem}
\newtheorem{claim}[theorem]{Claim}
\newtheorem{corollary}[theorem]{Corollary}
\newtheorem{lemma}[theorem]{Lemma}
\newtheorem{proposition}[theorem]{Proposition}

\theoremstyle{definition}

\newtheorem{definition}[theorem]{Definition}
\newtheorem{example}[theorem]{Example}
\newtheorem{remark}[theorem]{Remark}

\numberwithin{equation}{section}
\numberwithin{theorem}{section}

\begin{document}

\title{Location Games on Networks:\\
Existence and Efficiency of Equilibria\thanks{
This work was partially supported by PRIN 20103S5RN3, Galileo G15-30, and MOE2013-T2-1-158. Part of this research was completed while Ga\"etan Fournier was visiting Singapore University of Technology and Design in 2013 and 2014 and both authors were visiting the Institute for Mathematical Sciences, National University of Singapore in 2015. Support from the ANR Labex IAST is gratefully acknowledged.
}
}
\author{Ga\"etan Fournier\\
IAST\\
Manufacture des Tabacs\\
26 All\'ee de Brienne\\
31000 Toulouse\\
\texttt{fournier.gtn@gmail.com}
\and
Marco Scarsini\thanks{This  author is a member of GNAMPA-INdAM.}\\
Dipartimento di Economia e Finanza\\
LUISS\\
Viale Romania 32\\
00197 Roma, Italy\\
\texttt{marco.scarsini@luiss.it}
} 

\maketitle

\begin{abstract}
We consider a game where a finite number of retailers choose a location, given that their potential consumers are distributed on a network. Retailers do not compete on price but only on location, therefore each consumer shops at the closest store. We show that when the number of retailers is large enough, the game admits a pure Nash equilibrium and we construct it. We then compare the equilibrium cost borne by the consumers with the cost that could be achieved if the retailers followed the dictate of a benevolent planner. We perform this comparison in term of the  price of anarchy, i.e., the ratio of the worst equilibrium cost and the optimal cost, and the  price of stability, i.e., the ratio of the best equilibrium cost and the optimal cost. We show that, asymptotically in the number of retailers, these  ratios are bounded by two and one, respectively.

\end{abstract}

\noindent
\emph{MSC Subject Classification}: Primary	91A43; secondary 	91A06.

\noindent
\emph{OR/MS classification}: Games/noncooperative.

\noindent
\emph{Keywords}: Price of anarchy, price of stability, location games on networks, Hotelling games, pure equilibria, large games.

\section{Introduction}\label{se:intro}

\subsection{The problem}

In his seminal paper \citet{Hot:EJ1929} considers duopoly models, where two retailers compete by choosing a location and a price. 
The article is extremely rich in modeling, motivation, and examples in different areas. 
The most popular model considered by Hotelling involves two retailers who want to sell a homogeneous product to consumers who are uniformly distributed on a segment and make their purchase decision based on transportation costs and price of the product. The two retailers first simultaneously choose their location on the interval and then simultaneously choose the price of the product that they sell with the goal of maximizing their profit. 
Hotelling claims that, if transportation costs are linear, then a  \emph{principle of minimum differentiation} holds, that is, the only equilibrium is achieved when both retailers locate in the middle of the segment. 
\citet{dAsGabThi:E1979} show that there is a flaw in Hotelling's argument and, since payoffs are discontinuous, the principle of minimum differentiation actually does not hold.
This shows that  models where the retailers can choose both location and price are in general difficult to deal with.
 
Several variations of the model have been considered. 
Some of them assume that the price is exogenous and is the same for every retailer. 
In this case the competition is based only on retailers' location.
This model applies for instance to shops that sell products whose price is exogenously determined, for instance newsstands, pharmacies, or franchises of different types of services and products, e.g., brand clothes. 
Models that involve only location have been used also in political science \citep[see, e.g.,][]{Dow:HR1957} to explain why in a two-party system the parties tend to adopt  similar political platforms. 
Some interesting generalizations assume that set of feasible locations for the retailers is not necessarily a segment.
Our contribution goes in this direction.

\subsection{Our contribution}

In this paper we consider a model where consumers are uniformly distributed on a network and a finite number of retailers sell a unique homogeneous good and decide where to set shop. 
They can choose any location on the network and their choice is not limited to the vertices. 
The number of retailers and the price of the good they sell are exogenous. Each consumer buys the same amount of goods and decides to shop at the closest shop. 
Hence, the cost that a consumer incurs is the distance that he needs to travel to buy the good and the utility of the retailer is the share of the market that she can conquer, i.e., the mass of consumers that patronize her shop. 
This defines a normal form game, called \emph{location game} where the players are the retailers. 

In the first part of the paper we 
provide 
conditions for the existence of pure Nash equilibria in location games. 
In particular we show that for every possible network there exists a threshold $\bar{n}$ such that, whenever the number of retailers exceeds $\bar{n}$, the game admits a pure Nash equilibrium. 
The proof of this result is constructive.
We also consider special examples of networks for which more precise results can be obtained. 

In the second part of the paper we turn to analyze how efficient the equilibria of location games are. 
This is usually achieved with the price of anarchy, i.e., the ratio of the optimum social payoff and the social payoff induced by the worst Nash equilibrium.
\citep{KouPap:STACS1999}. 
The price of stability is constructed in a similar way by replacing the worst with the best Nash equilibrium \citep{SchSti:P14ACMSIAM2003}.
In both cases the social payoff is the sum of the payoffs achieved by all players in the game.

A location game is a constant-sum game, therefore, from the retailers' point of view, the social payoff is the same for every possible strategy profile. As a consequence, in their standard form, both the price of anarchy and the price of stability are equal to one.

From the consumers' point of view, a location game is not constant-sum, that is, the retailers' decisions affect the cost incurred by the consumers both individually and socially. 
Therefore something interesting can be said by examining the efficiency of equilibria from the consumers' viewpoint, i.e., by considering the total transportation cost incurred by all consumers to reach the closest shop.
This is the function that we use to compute the price of anarchy and the price of stability of a location game. 
We  prove that, for every network, as the number $n$ of retailers increases, there is a bound on the price of anarchy that is asymptotically not larger than $2$ and a bound on the price of stability that is asymptotically~$1$. 
We show that the bound on the price of anarchy is only asymptotic, i.e., for finite values of $n$ the price of anarchy can be larger than $2$. Moreover the convergence is not monotone: there 
exist
networks for which the price of anarchy is infinitely often strictly larger than $2$ and infinitely often strictly smaller than $2$.
The results on the  price of anarchy and on the  price of stability are proved using majorization techniques.

\subsection{Related literature}

As mentioned before, in general finding equilibria in models where both locations and prices are endogenous is a hard problem. 
For instance \citet{dAsGabThi:E1979} use quadratic transportation costs to obtain an equilibrium with two retailers. 
Interestingly enough, in this equilibrium the two retailers want to locate as far as possible from each other. 

To overcome these issues, some papers consider pure location models with exogenous prices.
For instance, \citet{EatLip:RES1975} study pure Nash equilibria for location games on the segment for an arbitrary number $n$ of retailers and they show that, when consumers are uniformly distributed, an equilibrium exists for any $n\neq 3$. 
A similar phenomenon, where equilibria exist for small and large values of $n$, but not for intermediate values, will be studied in Subsection~\ref{suse:star} for consumers distributed on a star.
Some papers consider consumers distributed on a plane. For instance, 
\citet{Los:YUP1954} and \citet{BolSte:JET1972} show that a strategy profile that splits the plane into hexagonal domains of attraction is socially optimal for the consumers.
\citet{Sal:BJE1979} considers a model with two retailers, where  consumers are distributed on a circle. This assumption simplifies the analysis with respect to the case of the interval, by eliminating the corner effects. 
\citet{EisLap:TS1993} find pure Nash equilibria when three retailers locate their shops on a tree.

Some papers focus on mixed equilibria in  location models.
\citet{Sha:JIE1982} finds a mixed equilibrium for the case of three retailers on the segment.
\citet{OsbPit:IER1986} study mixed Nash equilibria for location games on a segment under general assumptions on the consumers' distribution and they show that, as the number of retailers increases, the mixed strategy in the symmetric equilibrium of the game tends to mimic the distribution of the consumers.
A similar phenomenon is studied in \citet{NunSca:ETB2016} for retailers whose finite choice set is a subset of a general compact metric space.

Closer to the scope of our paper, a few authors consider location models on a graph. For instance \citet{DurTha:ESA2007} and \citet{MavMonPap:Springer2008} study a class of games called \emph{Voronoi games} where players choose a vertex $v$ in a finite graph and the payoff of each player is the Voronoi cell of $v$, that is, the set of vertices that is closer to $v$ than to any other chosen vertex. In our language this would correspond to a setting where both retailers and consumers live only on the vertices of a network and all the edges have the same length. For cyclic graphs \citet{MavMonPap:Springer2008} find bounds for the price of anarchy. 
Their result is similar to ours in the sense that the price of anarchy is not computed in terms of  retailers' payoffs, but rather in terms of consumers' costs. 
The big difference is that their game is finite, since the action set for players in their game is the finite set of vertices.
Our model is the same as the one studied by \citet{Pal:mimeo2011}. 
In his paper he finds conditions for existence of pure equilibria for location games on a graph. 
We make use of several of his intermediate results to prove our existence theorem and we fix a gap in his proof. 
The details of the similarities and differences between our and his existence proof will be described in  Section~\ref{se:existence}.

In an interesting paper \citet{HeiSoe:Tinb2014} consider a model where consumers are uniformly distributed on a graph and two retailers choose prices but not location. They overcome the difficulties of dealing with a network by focusing on prices and keeping locations fixed, whereas we do the opposite: we assume that the price is exogenous and we focus on location.

As mentioned before, we measure inefficiency of equilibria with the price of anarchy and price of stability. These measures were introduced by  
\citet{KouPap:STACS1999} and \citet{SchSti:P14ACMSIAM2003}, respectively and given these names by \citet{Pap:PACM2001} and 
\citet{AnsDasKleTarWexRou:SIAMJC2008}, respectively.
The papers by \citet{Vet:FOCS2002} and \citet{MavMonPap:Springer2008} are two examples where the social cost used to compute the price of anarchy is not the sum of the costs of the individual players. The same happens here, where the game is a payoff game for the retailers and a cost game for the consumers. 

Our efficiency results use majorization techniques \citep[see, e.g.,][]{MarOlkArn:Springer2011}. 
Although majorization is a very well-known tool in various areas of mathematics, probability, statistics, and, more recently, economics, to the best of our knowledge, it has not been used in game theory.

\subsection{Organization of the paper}

In Section~\ref{se:model} the model is introduced. Section~\ref{se:existence} proves existence of pure equilibria for location games with a large number of  players. Section~\ref{se:efficiency} deals with efficiency of these equilibria. 
Several examples are considered in Section~\ref{se:examples}.

\section{The model}\label{se:model}

We start by providing a formal definition of network. Then we describe the normal form location game played on this network.

\subsection{The network}\label{suse:network}

Consider a graph $(V,E)$, where $V$ is a finite set of \emph{vertices} and $E$ is a finite set of \emph{edges}. 
If the edge $e$ joins the vertices $u$ and $v$, we use the notation $(u,v):=e$.
Based on $(V,E)$, we construct a set that we endow with a distance and a measure.
First we associate to each edge $e$ a value $\lambda(e)>0$, called the \emph{length} of $e$. 
We want to treat each edge $e$ like an interval of length $\lambda(e)$, so, for any edge $e=(u,v) \in E$ we call
\begin{equation*}
(u,v,\alpha)=\alpha u + (1-\alpha) v
\end{equation*}
the convex combination of $u$ and $v$ with weights $\alpha$ and $(1-\alpha)$. 
The point $(u,v,0)$ is identified with the vertex $v$ and the point $(u,v,1)$ is identified with the vertex $u$.
Each point on an edge is defined by two different triplets since $(u,v,\alpha) = (v,u,1-\alpha)$. 
If $x_{1}=(u,v,\alpha_{1})$ and $x_{2}=(u,v,\alpha_{2})$, then we define the \emph{interval}
\begin{equation*}
[x_{1},x_{2}]:=\{(u,v,\eta) : \min(\alpha_{1},\alpha_{2}) \le \eta \le \max(\alpha_{1},\alpha_{2})\}.
\end{equation*}
If $e=(u,v)$, with an abuse of language we  use the notation $e$ also for the interval $[u,v]$.
Consider the set  
\begin{equation}\label{eq:network}
S:=\{(u,v,\eta) : u,v \in V, (u,v) \in E, \eta \in [0,1]\}.
\end{equation}
Now we endow  $S$ with a measure $\lambda$ as it follows. First $\lambda$ is defined on intervals: if $x_{1}=(u,v,\alpha_{1})$ and $x_{2}=(u,v,\alpha_{2})$, then 
\begin{equation*}
\lambda([x_{1},x_{2}]) = \lambda ([u,v]) \times  \vert \alpha_{2} - \alpha_{1} \vert.
\end{equation*} 
Then $\lambda$ is additively extended to the $\sigma$-field generated by the intervals. 

We are now ready to define a metric $d$ on $S$ that is coherent with $\lambda$. For any two points $x,y \in S$, the distance $d(x,y)$ is the measure $\lambda$ of the shortest path that joins $x$ and $y$. 
From now on, for the sake of concision, we call $S$ the metric measurable \emph{network} $(S, d, \lambda)$ and we say that $S$ is generated by $(V,E,\lambda)$.

We call \emph{leaf} a vertex $v \in V$ such that $\degree(v)=1$.
The network $S$ generated by a graph $(V,E,\lambda)$ is equivalent to a network generated by a sub-graph whose vertices have degree different from $2$. This subgraph can be obtained by performing this operation: whenever the vertex $u$ has degree two, delete it and replace  $[v,u]$ and $[u,w]$ with  $[v,w]$ so that $\lambda([v,w])=\lambda([v,u])+\lambda([u,w])$. Therefore we will always assume that $V$ contains no vertices of degree $2$. We then extend the definition of the function degree from $V$ to $S$ by assuming that $\degree(x)=2$ for all $x\in S \setminus V$.

\subsection{Retailers and consumers}\label{suse:sellers}

We consider a situation where each of $n$ retailers has to decide where to locate her shop on a network, given that a continuum of consumers is uniformly distributed on the network according to $\lambda$ and each consumer patronizes a shop in the closest location.

Ties may arise and they are solved as follows. Consider the set $A$ of consumers that are equally distant from $k$ different locations having each at least one shop. Then we assume that for each of these $k$ locations, $\lambda(A)/k$ consumers go to that location. Moreover if one of the $k$ locations has $h$ shops, then a fraction $\lambda(A)/(hk)$ patronizes each shop of this location.
Basically the network is decomposed into domains of attraction of different retailers's locations and then within each domain of attraction retailers in the same location split the consumers equally. Some parts of the network can belong to different domains of attraction, as the following example shows. 

\begin{example}
Consider the network in Figure~\ref{fi:attraction} with seven players. Assume that $\lambda(e_{4})=\lambda(e_{5})$ and that two retailers are located in $u$ and five retailers are located in $v$. All the points in $e_{8}$ are equally distant from $u$ and $v$. 
Therefore the retailers in $u$ jointly attract all the consumers on the  solid edges  plus half of the consumers on $e_{8}$. 
The retailers in $v$ jointly attract all the consumers on the dashed edges  plus the remaining half of the consumers on $e_{8}$. That is, each player in $u$ attracts the following quantity of consumers
\[
\frac{1}{2}\left(\lambda(e_{1})+\lambda(e_{2})+\lambda(e_{3})+\lambda(e_{4})+\frac{1}{2}\lambda(e_{8})\right)
\]
and each player in $v$ attracts the following quantity of consumers
\[
\frac{1}{5}\left(\lambda(e_{5})+\lambda(e_{6})+\lambda(e_{7})+\frac{1}{2}\lambda(e_{8})\right).
\]
This example shows that the situation of a location game on a general network is more complicated than the classical case of a game on a circle or a segment. The fact that a set of positive measure may be equidistant from two points imposes some extra care in the definition of the domain of attraction of retailers.
\end{example}

\begin{center}

\begin{figure}[H]
\centering
\tikzstyle{leaf}=[rectangle,fill,scale=0.2]
\tikzstyle{gleaf}=[rectangle,fill=white,scale=0.2]
\tikzstyle{rleaf}=[rectangle,fill=white,scale=0.2]
\tikzstyle{single}=[circle,fill=white,scale=0.4]
\tikzstyle{double}=[rectangle,fill=red,scale=0.8]
\tikzstyle{rple}=[circle,fill=blue,scale=0.8]
\newdimen\mydim
\newdimen\myeps
\begin{tikzpicture}
\mydim=3cm;
\myeps=0.01mm;
\begin{scope}[thick]
   \draw node(l0) [rleaf] at (0,\mydim) {};
   \draw node(l2) [rleaf] at (0,-\mydim) {};
   \draw node(l3) [gleaf] at (10,-\mydim) {};   
   \draw node(l5)  at (5,-\mydim) {};   
   \draw node(l6)  at (2,0) {};
   \draw node(l7)  at (8,0) {};
   \draw node(l8)  at (5,0) {};
   \draw node(19)  at (5.02,0) {};
   \draw node(20)  at (5.02,-\mydim) {};    
   \draw node(21)  at (4.98,0) {};
   \draw node(22)  at (4.98,-\mydim) {};    
   \draw node(23) [rleaf] at (0,0) {};   
   \draw node(24) [gleaf] at (10,\mydim) {}; 
   \draw node(25) [gleaf] at (10,-\mydim) {};   
   \draw  [red] (l0) -- node[left, black] {$e_{1}$} (l6.center);
   \draw  [red] (l2) -- node[left, black] {$e_{3}$} (l6.center);
   \draw  [red] (23) -- node[above, black] {$e_{2}$} (l6.center);   
   \draw  [densely dashed, blue] (24) -- node[right, black] {$e_{6}$} (l7);
   \draw  [densely dashed, blue] (25) -- node[right, black] {$e_{7}$} (l7);
   \draw  [red] (l6.center) -- node[above, black] {$e_{4}$} (l8.center);
   \draw  [densely dashed, blue] (l7) -- node[above, black] {$e_{5}$} (l8.center);
   \draw  [densely dashed, blue] (19.center) -- node[right, black] {$e_{8}$} (20.center);
   \draw  [red] (21.center) -- (22.center);
   \draw node(a) [double] at (11.5, 0.4) {};
   \draw (11.5,0.4) node [right]  {\ {$2$ players in $u$}}; 
   \draw node(b) [rple] at (11.5, -0.4) {};
   \draw (11.5,-0.4) node [right]  {\ {$5$ players in $v$}};
   \draw node(c) [double] at (2,0) {} ;
   \draw node(d) [rple] at (8,0) {} ; 
   \draw (2,-5pt)   node [below] {$u$};
   \draw (8,-5pt)   node [below] {$v$};
   \end{scope}
   \end{tikzpicture}
~\vspace{0cm} \caption{\label{fi:attraction} Domains of attraction when $\lambda(e_{4})=\lambda(e_{5})$.}
\end{figure}
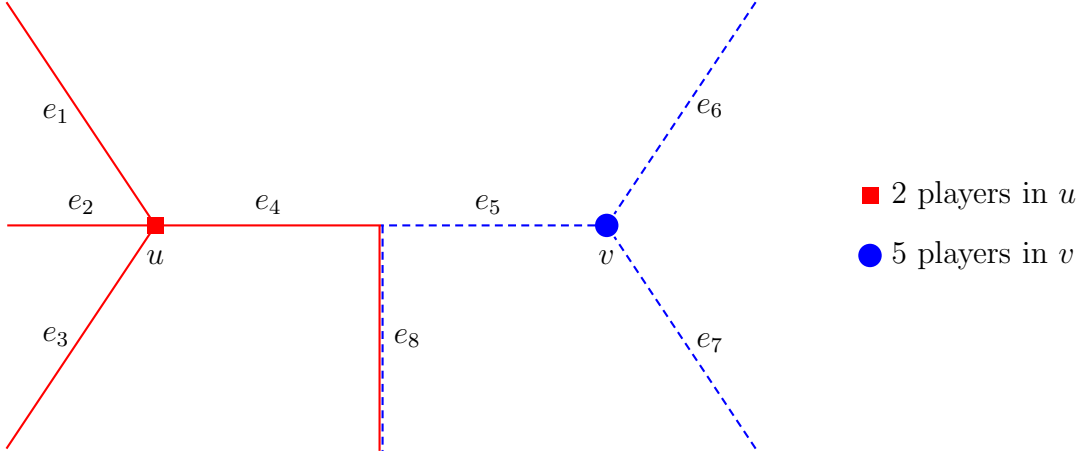
\end{center}

\subsection{The normal-form game}\label{suse:game}

We model this 
situation 
as a normal form game with a finite set $N=\{1, \dots, n\}$ of players having all the same action set $S$. Each player represents a retailer, whose payoff is the measure of the set of consumers who shop at her store. 
In order to formally define the players' payoffs, we need to introduce the following quantities.
Given a strategy profile $\boldsymbol{x}=(x_{1}, \dots, x_{n})\in S^{n}$, call $\Psi(\boldsymbol{x})$ the set of all locations that are occupied by some shop under the strategy profile $\boldsymbol{x}$, that is, the set of all $y\in S$ for which there exists $i\in N$ such that $y=x_{i}$. Given a set $A$, call $\card(A)$ its cardinality.
It can happen that $\card(\Psi(\boldsymbol{x})) < n$, since several players can choose the same location. 
For $K \subset \Psi(\boldsymbol{x})$ define
\begin{equation}\label{eq:YK}
\begin{split}
Y_{K} &= \{y\in S: d(y,x_{i}) = d(y,x_{j}) \text{ for all pairs } x_{i},x_{j}\in K \\
&\qquad\text{ and } d(y,x_{i}) < d(y,x_{\ell})  \text{ for all }x_{i}\in K, x_{\ell} \not\in K\}.
\end{split}
\end{equation}

The set $Y_{K}$ is the subset of consumers who are indifferent between all retailers in $K$ and strictly prefer retailers in $K$ to any other retailer outside $K$. Obviously for some choices of $K$ the set $Y_{K}$ can be empty.
The payoff of player $i\in N$ under the strategy profile $\boldsymbol{x}$ is
\begin{equation}\label{eq:payoff}
\rho_{i}(\boldsymbol{x}) = \frac{1}{\card(\{j\in N: x_{j}=x_{i}\})}\sum_{K\subset N} \mathds{1}_{x_{i}\in K} \frac{\lambda(Y_{K})}{\card(K)}.
\end{equation}

The above defined game is called \emph{location game} on $S$ with $n$ players and is denoted by $\mathcal{L}(n, S)$. 
A strategy profile $\boldsymbol{x}^{*}$ is a \emph{pure Nash equilibrium} of the game $\mathcal{L}(n, S)$ if for all $i \in N$ and for all $x_{i}\in S$ we have
\[
\rho_{i}(\boldsymbol{x}^{*}) \ge \rho_{i}(x^{*}_{1}, \dots, x^{*}_{i-1}, x_{i}, x^{*}_{i+1}, \dots, x^{*}_{n}). 
\]
For the sake of simplicity, in the rest of the paper we will use the term equilibrium to indicate a pure equilibrium.

\section{Existence of equilibria}\label{se:existence}

In this section we deal with existence of equilibria in location games. It is well-known that some location games do not admit equilibria (see, e.g., Proposition~\ref{pr:zeroone} below). Since the players' action spaces  are infinite and  their payoff functions are not continuous, no general known result can be used to prove existence. Therefore, more specific arguments will have to be employed, which rely on some structural properties of equilibria in location games.
The following theorem shows that a location game on any network $S$ always admits a pure Nash equilibrium, provided the number of players is large enough.

\begin{theorem}\label{th:existencebign}
For an arbitrary $S$, there exists $\bar{n}\in \mathbb{N}$ such that for every $n \ge \bar{n}$, the game  $\mathcal{L}(n, S)$ admits a pure Nash equilibrium.
\end{theorem}

We will show in the proof of Theorem~\ref{th:existencebign} that we can take
\begin{equation}\label{eq:nbardef}
\bar{n} = 3 \card(E)+\sum_{e\in E}\left\lceil\frac{5\lambda(e)}{\min_{e' \in E}\lambda(e')}\right\rceil,
\end{equation}
where $\lceil x\rceil$ is the ceiling of $x$.

A similar result with a different bound was proposed in an unpublished paper by \citet{Pal:mimeo2011}, where he provides an algorithm as a constructive proof of the existence of equilibrium. However, his algorithm provides a profile of location where the number of players on an edge depends only on its length, hence, all the edges with the same length must have the same number of players. As a consequence, his proof is incomplete: for instance, in a graph where all the edges have the same length, his construction holds only for a number of players that is proportional to the number of edges.

Given this issue, we provide a complete proof of Theorem~\ref{th:existencebign}. We use several steps of P\'alv\"olgyi's construction and fill the gap in his proof as detailed in Remark~\ref{re:Palvolgy} below. 
Our proof is constructive. 
First we show that to be an equilibrium of a location game a strategy profile must satisfy several necessary conditions.
These necessary conditions  provide a solid structure for equilibria in location games and are the building blocks in the construction of our equilibrium for games with a large number of players.

\subsection{Proofs}

We first state some properties of equilibria that will be useful both to prove the existence of equilibria and to compute their efficiency.

Given a graph $(V,E)$ without vertices of degree $2$, define:
\begin{align}
V_{I}&=\{v \in V : \degree(v)\geq 3\}, \label{eq:VI}\\
V_{L}&=\{v \in V : \degree(v)=1\}, \label{eq:VL}\\
E_{IL} &= \{e \in E : e= (v,w), v \in V_{I}, w \in V_{L}\}, \label{eq:EIL}\\
E_{LL} &= \{e \in E : e= (v,w), v,w \in V_{L}\}, \label{eq:ELL}\\
E_{II} &= \{e \in E : e= (v,w), v,w \in V_{I}\}. \label{eq:EII}
\end{align}

Our proof of the existence of equilibria in games with a large number of players provides an equilibrium that satisfies the following key condition.

\begin{definition}\label{de:VP}
A strategy profile $\boldsymbol{x}$ satisfies the \emph{vertex property} if, for all $v \in V_{I}$, as defined in \eqref{eq:VI}, there exists $i \in N$ such that $x_{i} = v$.
\end{definition}

Notice that, by Lemma~\ref{le:largen} below, the vertex property is always satisfied by equilibria of games with a large enough number of players.

\begin{lemma}\label{le:playersleaf}
Consider an edge $e$ that connects a leaf $v$ and a vertex $w$.
Let the equilibrium $\boldsymbol{x}^{*}$ of $\mathcal{L}(n, S)$ satisfy the vertex property.
If, under this equilibrium, the closest player to the leaf $v$ is in the interior of $e$, then she cannot be alone.
\end{lemma}

\begin{proof}
Assume, by contradiction, that the closest player to the leaf $v$ is in the interior of $e$ and is alone. Then she would profitably deviate by moving away from the leaf $v$, for as long as she does not overcome the next player on $e$.
\end{proof}

\begin{lemma}\label{le:playersdegree}
Take any point $w \in S$. If the equilibrium $\boldsymbol{x}^{*}$ of $\mathcal{L}(n, S)$ satisfies the vertex property, then  
$\card\{i \in N : x_{i}^{*} = w\} \le \degree(w)$.
\end{lemma}

Lemma~\ref{le:playersdegree} implies that no more than two players can share the same location in the interior of any edge. Only the vertices can have more than two players, but never more than the degree of the vertex.

\begin{proof}[Proof of Lemma~\ref{le:playersdegree}]
We start with the case where $w \in V$.
Consider an equilibrium $\boldsymbol{x}^{*}$ of $\mathcal{L}(n, S)$ that satisfies the vertex property. Assume, by contradiction, that there exists a point $w$ such that 
\[
\card\{i \in N : x_{i}^{*} = w\} =: k > \degree(w). 
\]
Consider the set $\mathcal{I}(w)$ of all the edges that are incident on $w$. We can partition $\mathcal{I}(w)$ into a set $\mathcal{N}(w)$ of edges that have a leaf and no player on them, except in $w$, and its complement $\mathcal{P}(w)$. 
For $e\in \mathcal{P}(w)$ we call $c(e)$ the location of the closest player to $w$ on the edge $e$, who exists, since every vertex that is not a leaf has at least one player. 
For $e\in \mathcal{I}(w)$ we define
\begin{equation}\label{eq:delta}
\delta(w,e) =
\begin{dcases}
\frac{d(w,c(e))}{2}& \text{if }e \in \mathcal{P}(w),\\
\lambda(e)& \text{if }e \in \mathcal{N}(w).
\end{dcases}
\end{equation}
Therefore $\delta(w,e)$ represents the mass of consumers on the edge $e$ who shop at location $w$.
Given that $\card(\mathcal{I}(w)) = \degree(w)$,  each player in $w$ gains 
\[
\frac{1}{k}\sum_{e\in\mathcal{I}(w)} \delta(w,e)\le \frac{\card(\mathcal{I}(w))}{k} \max_{e\in\mathcal{I}(w)} \delta(w,e) < \max_{e\in\mathcal{I}(w)} \delta(w,e).
\]
Then, for $\varepsilon$ small enough, a player who moves by $\varepsilon$ from $w$ in the direction of $\arg\max_{e}\delta(w,e)$ enjoys a profitable deviation.

If $w \not\in V$, then we can make it a vertex by splitting the edge that contains $w$ into two edges incident on $w$. The previous argument goes through.
\end{proof}

\begin{definition}
Consider a strategy profile $\boldsymbol{x} \in S^n$ and $w \in S$. If $\card\{i \in N : x_{i} = w\}=\degree(w)$, then any player $j$ such that $x_{j}=w$ is called \emph{$\boldsymbol{x}$-balanced}, and the location $w$ is said to be \emph{$\boldsymbol{x}$-saturated}.
\end{definition}

\begin{lemma}\label{le:degreecost}
For $n \ge 2$, if $\boldsymbol{x}^{*}$ is an equilibrium of $\mathcal{L}(n, S)$ and player $i$ is $\boldsymbol{x}^{*}$-balanced, then
\[
\rho_{i}(\boldsymbol{x}^{*}) \le \rho_{j}(\boldsymbol{x}^{*}) \quad \text{for all } j \in N.
\]  
\end{lemma}

\begin{proof}
Let $i$ be a $\boldsymbol{x}^{*}$-balanced player. Then 
\[
\rho_{i}(\boldsymbol{x}^{*}) = \frac{1}{\degree(x^{*}_{i})}\sum_{e\in\mathcal{I}(x^{*}_{i})} \delta(x^{*}_{i},e) \le \max_{e\in\mathcal{I}(x^{*}_{i})} \delta(x^{*}_{i},e),
\]
where $\delta$ is defined in equation \eqref{eq:delta}. Assume, \emph{ad absurdum}, that there exists a player $j$ such that $\rho_{j}(\boldsymbol{x}^{*}) < \rho_{i}(\boldsymbol{x}^{*})$, then player $j$ could deviate on the edge 
$\arg\max_{e} \delta(x^{*}_{i},e)$ at a distance $\varepsilon$ from $x^{*}_{i}$ and gain 
\[
\max_{e\in\mathcal{I}(x^{*}_{i})} \delta(x^{*}_{i},e) - \frac{\varepsilon}{2} > \rho_{j}(\boldsymbol{x}^{*}). \qedhere
\]
\end{proof}

\begin{corollary}\label{co:equalperiph}
In equilibrium all balanced players get the same payoff.
\end{corollary}

\begin{proof}
Just consider two balanced players and apply Lemma~\ref{le:degreecost}  to show that the payoff of each one of them is smaller or equal than the payoff of the other.
\end{proof}

\begin{corollary}\label{co:equaldelta}
Consider an  equilibrium $\boldsymbol{x}^{*}$  of $\mathcal{L}(n, S)$ that satisfies the vertex property. Then there exists $\xi>0$ such that for every $\boldsymbol{x}^{*}$-saturated location $w$ and every $e\in \mathcal{I}(w)$, we have $\delta(w,e) = \xi$. Moreover every player on an  $\boldsymbol{x}^{*}$-saturated location has a payoff equal to $\xi$.
\end{corollary}

\begin{proof}
We start by showing that for each $\boldsymbol{x}^{*}$-saturated location $w$ we have 
\[
\delta(w,e) = \delta(w,e') \quad\text{for all } e, e' \in \mathcal{I}(w).
\]
If this were not true, i.e., if we had
\[
\max_{e\in\mathcal{I}(w)} \delta(w,e) > \frac{1}{\card(\mathcal{I}(w))} \sum_{e\in\mathcal{I}(w)} \delta(w,e),
\]
then one of the players in $w$ could profitably deviate of $\varepsilon$ on the edge 
$\arg \max_{e} \delta(w,e)$.
So each player in $w$ has the same payoff that we denote $\xi_w$. Now, using Corollary~\ref{co:equalperiph},  for every pair $v,w$ of $\boldsymbol{x}^{*}$-saturated locations  we have $\xi_v = \xi_w =: \xi$.
\end{proof}

\begin{definition}\label{def:redundant}
In an equilibrium $\boldsymbol{x}^{*}=(x_{1}^{*},\dots,x_{n}^{*})$ of the game $\mathcal{L}(n,S)$, player $i \in \{1,\dots,n\}$ is a \emph{redundant player} if the profile of location $(x_{1}^{*},\dots,x_{i-1}^{*},x_{i+1}^{*},\dots,x_{n}^{*})$ is an equilibrium of the game $\mathcal{L}(n-1,S)$.
\end{definition}

Our goal is to construct an equilibrium of the game $\mathcal{L}(n, S)$ for all $n \ge \bar{n}$, where $\bar{n}$ is defined as in \eqref{eq:nbardef}. 
We pick an $n \ge \bar{n}$ and first we find an equilibrium for a game  $\mathcal{L}(n', S)$, where $n'$ is slightly larger than $n$, in a way that the next proposition will make precise. 
Then we prove that there exists enough redundant players to transform this equilibrium with $n'$ players into an equilibrium of a game with exactly $n$ players.

\begin{proposition}\label{pr:equilibriumnprime}
For a given $S$ there exists $\bar{n} \in \mathbb{N}$ such that for all $n \ge \bar{n}$ there exists $n' \in \mathbb{N}$ for which
\begin{enumerate}[{\rm (a)}]
\item
$n \le n' \le n+\card(E)$,

\item
the game $\mathcal{L}(n', S)$ admits a Nash equilibrium.
\end{enumerate}
\end{proposition}

The proof of this proposition requires the following lemma. 
Given $S$, define $f:\mathbb{R}_{+}\to \mathbb{N}$ as follows:

\begin{equation}\label{eq:definitionf}
f(z) = 3\card(E) + \sum_{e\in E} \left\lceil \frac{\lambda(e)}{2z} \right\rceil.
\end{equation}

The quantity $f(\xi)$ represents the number of players on the network in our equilibrium, as a function of the quantity $\xi$ defined in Corollary~\ref{co:equaldelta}.

\begin{lemma}\label{le:graphf}
For all $n \ge 4 \card(E)$, there exist $\underline{\xi}, \overline{\xi} \in \mathbb{R}_{+}$ such that 
\[
n \le f(z) \le n + \card(E), \quad \text{for all } z \in[\underline{\xi}, \overline{\xi}).
\]
\end{lemma}

\begin{proof}
The function $f$ in \eqref{eq:definitionf} is defined as the sum of a constant and $\card(E)$ terms each one of which is a piecewise constant, weakly decreasing, and right continuous function with jumps of magnitude $1$.
Therefore, for all $z_{0}>0$, 
\[
0 \le \lim_{z \to z_{0}^{-}} f(z)-f(z_{0}) \le \card(E).
\]
Moreover, we have
\begin{align*}
\lim_{z \to \infty} f(z) &= 3 \card(E) + \card(E) = 4 \card(E), \\
\lim_{z \to 0^{+}} f(z) &= +\infty.  \qedhere
\end{align*}
\end{proof}

\begin{proof}[Proof of Proposition~\ref{pr:equilibriumnprime}]
Let $f$ be defined as in \eqref{eq:definitionf}. Then
\begin{equation}\label{eq:barn}
\bar{n} := f\left(\frac{\min_{e \in E} \lambda(e)}{10}\right) 
= 3 \card(E)+\sum_{e\in E}\left\lceil\frac{5\lambda(e)}{\min_{e'}\lambda(e')}\right\rceil.
\end{equation}
Take $n \ge \bar{n}$.
By Lemma~\ref{le:graphf} there exists an interval $[\underline{\xi}, \overline{\xi})$ such that, for $\xi \in[\underline{\xi}, \overline{\xi})$ we have $f(\xi)=n'$, with $n\le n' \le n+\card(E)$.

We choose $\xi =  \underline{\xi}$ and construct a Nash equilibrium $\boldsymbol{x}^{*}$ of $\mathcal{L}(n', S)$. This notation is coherent with the previous definition of $\xi$ because players in saturated locations get a payoff equal to $\xi$.
To achieve the equilibrium, we position players on the edges of $S$ as follows.

\medskip
\noindent First case: $e \in E_{IL}$. 
If $e=(v,w)$, with $w \in V_{L}$, then, under $\boldsymbol{x}^{*}$, the number of players on $[v,w]$ is set to
\[
p(e):=\degree(v) +  \left\lceil \frac{\lambda(e)}{2\xi} \right\rceil + 2.
\]
Out of these players, $\degree(v)$ will be in $v$, and the remaining will be as in Figure~\ref{fi:EL}. 
Therefore the edge $e$ is split into three intervals of length $2\xi$, one interval of length $\xi$ and $(p(e)-\degree(v)-5)$ intervals of length $\alpha(e)\xi$, where $\alpha(e)$ is a parameter such that $1 \le \alpha(e) \le 2$. 
Taking into account the number of players on $e$, the length $\lambda(e)$, and the number of intervals of length $\alpha(e)\xi$, we have
\begin{equation}\label{eq:alphaEL}
\alpha(e)= \frac{\lambda(e) - 7\xi }{\xi \left\lceil \frac{\lambda(e)}{2\xi} \right\rceil - 3\xi}.
\end{equation}
\bigskip

\begin{center}

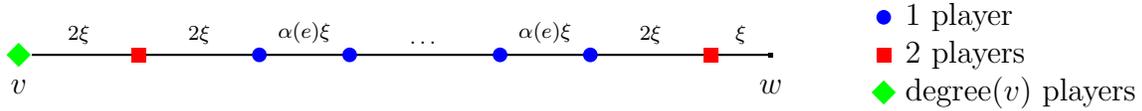
\begin{figure}[H]
\centering
\tikzstyle{leaf}=[rectangle,fill,scale=0.2]
\tikzstyle{single}=[circle,fill=blue,scale=0.5]
\tikzstyle{double}=[rectangle,fill=red,scale=0.7]
\tikzstyle{rple}=[diamond,fill=green,scale=0.6]
\tikzstyle{kple}=[regular polygon,regular polygon sides=3,fill=brown,scale=0.5]

\begin{tikzpicture}
  \begin{scope}[thick]
   \draw node(l0) [rple] at (0,0) {};
   \draw node(l1) [leaf] at (10,0) {};
   
   \draw  (l0) -- (l1);
           
   \draw node(a) [single] at (11.5, 0.5) {};
   \draw (11.5,0.5) node [right]  {\ {$1$ player}}; 
       
   \draw node(b) [double] at (11.5, 0) {};
   \draw (11.5,0) node [right]  {\ {$2$ players}};

   \draw node(c) [rple] at (11.5, -0.5) {};
   \draw (11.5,-0.5) node [right]  {\ {$\degree(v)$ players}};   

   \draw (0,-5pt)   node [below] {$v$};
   \draw (10,-5pt)   node [below] {$w$};
         
   \draw node(c) [double] at (16/10,0) {} ;
   \draw (8/10,0)   node [above, font=\scriptsize] {$2\xi$};   

   \draw node(h) [single] at (32/10,0) {} ;
   \draw (24/10,0)   node [above, font=\scriptsize] {$2\xi$};      
   
   \draw node(d) [single] at (44/10,0) {} ;   
   \draw (38/10,0)   node [above, font=\scriptsize] {$\alpha(e)\xi$};
  
   \draw (54/10,0)   node [above, font=\scriptsize] {$\dots$};  

   \draw node(f) [single] at (64/10,0) {} ;
   \draw (70/10,0)   node [above, font=\scriptsize] {$\alpha(e)\xi$};

   \draw node(h) [single] at (76/10,0) {} ;
   \draw (84/10,0)   node [above, font=\scriptsize] {$2\xi$};   
      
   \draw node(n) [double] at (92/10,0) {} ;
   \draw (96/10,0)   node [above, font=\scriptsize] {$\xi$};   

   \end{scope}
   
\end{tikzpicture}
~\vspace{0cm} \caption{\label{fi:EL} Players in $e\in E_{IL}$.}
\end{figure}
\end{center}

\bigskip

\noindent Second case: $e \in E_{II}$. 
If $e=(v,w)$, with $v,w\in V_{I}$, then, under $\boldsymbol{x}^{*}$, the number of players on $[v,w]$ is set to
\[
\degree(v) + \degree(w) + \left\lceil \frac{\lambda(e)}{2\xi} \right\rceil +1.
\]
Out of these players, $\degree(v)$ will be in $v$, $\degree(w)$ will be in $w$, and the remaining will be as in Figure~\ref{fi:EI}. This implies that
\begin{equation}\label{eq:alphaEI}
\alpha(e)= \frac{\lambda(e) - 6 \xi}{\xi \left\lceil \frac{\lambda(e)}{2\xi} \right\rceil - 2\xi}.
\end{equation}
Notice that in this case the construction is still valid if we reverse the roles of $v$ and $w$.

\bigskip

\begin{center}

\begin{figure}[H]
\centering

\tikzstyle{leaf}=[rectangle,fill,scale=0.2]
\tikzstyle{single}=[circle,fill=blue,scale=0.5]
\tikzstyle{double}=[rectangle,fill=red,scale=0.7]
\tikzstyle{rple}=[diamond,fill=green,scale=0.6]
\tikzstyle{kple}=[regular polygon,regular polygon sides=3,fill=brown,scale=0.5]

\begin{tikzpicture}

  \begin{scope}[thick]
  
   \draw node(l0) [rple] at (0,0) {};
   \draw node(l1) [kple] at (10,0) {};
   
   \draw  (l0) -- (l1);
           
   \draw node(a) [single] at (11.5, 0.6) {};
   \draw (11.5,0.6) node [right]  {\ {$1$ player}}; 
       
   \draw node(b) [double] at (11.5, 0.2) {};
   \draw (11.5,0.2) node [right]  {\ {$2$ players}};

   \draw node(c) [rple] at (11.5, -0.2) {};
   \draw (11.5,-0.2) node [right]  {\ {$\degree(v)$ players}}; 
    
   \draw node(c) [kple] at (11.5, -0.6) {};
   \draw (11.5,-0.6) node [right]  {\ {$\degree(w)$ players}};   

   \draw (0,-5pt)   node [below] {$v$};
   \draw (10,-5pt)   node [below] {$w$};
         
   \draw node(c) [single] at (16/10,0) {} ;
   \draw (8/10,0)   node [above, font=\scriptsize] {$2\xi$};   
         
   \draw (22/10,0)   node [above, font=\scriptsize] {$\alpha(e)\xi$};

   \draw node(d) [single] at (28/10,0) {} ;
   \draw (36/10,0)   node [above, font=\scriptsize] {$\cdots$}; 
  
   \draw node(f) [single] at (44/10,0) {} ;
   \draw (50/10,0)   node [above, font=\scriptsize] {$\alpha(e)\xi$};  

   \draw node(f) [single] at (56/10,0) {} ;
   \draw (62/10,0)   node [above, font=\scriptsize] {$\alpha(e)\xi$};

   \draw node(h) [single] at (68/10,0) {} ;
   \draw (76/10,0)   node [above, font=\scriptsize] {$2\xi$};   
      
   \draw node(n) [double] at (84/10,0) {} ;
   \draw (92/10,0)   node [above, font=\scriptsize] {$2\xi$};   

   \end{scope}
   
\end{tikzpicture}
~\vspace{0cm} \caption{\label{fi:EI} Players in $e\in E_{II}$.}

\end{figure}
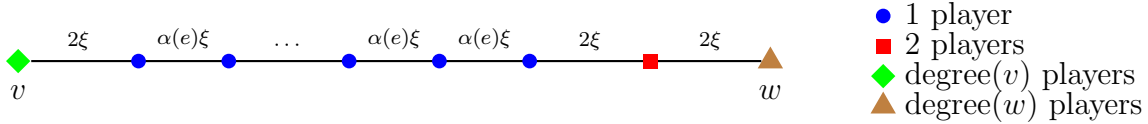
\end{center}

\bigskip

\noindent Third case: $e \in E_{LL}$. 
If $e=(v,w)$, with $v,w\in V_{L}$, then, under $\boldsymbol{x}^{*}$, the number of players on $[v,w]$ is set to
\[
\left\lceil \frac{\lambda(e)}{2\xi} \right\rceil +3.
\]
They will be located as in Figure~\ref{fi:ELL}. This implies that 
\begin{equation}\label{eq:alphaELL}
\alpha(e) = \frac{\lambda(e) - 8\xi}{\xi\left\lceil \frac{\lambda(e)}{2\xi} \right\rceil -4\xi}.
\end{equation}
As before, in this case the construction remains valid if we reverse the roles of $v$ and $w$.

\bigskip

\begin{center}

\begin{figure}[H]
\centering
\tikzstyle{leaf}=[rectangle,fill,scale=0.2]
\tikzstyle{single}=[circle,fill=blue,scale=0.5]
\tikzstyle{double}=[rectangle,fill=red,scale=0.7]
\tikzstyle{rple}=[diamond,fill=green,scale=0.6]
\tikzstyle{kple}=[regular polygon,regular polygon sides=3,fill=brown,scale=0.5]

\begin{tikzpicture}

  \begin{scope}[thick]
  
   \draw node(l0) [leaf] at (0,0) {};
   \draw node(l1) [leaf] at (10,0) {};
   
   \draw  (l0) -- (l1);
           
   \draw node(a) [single] at (11.5, 0.3) {};
   \draw (11.5,0.3) node [right]  {\ {$1$ player}}; 
       
   \draw node(b) [double] at (11.5, -0.3) {};
   \draw (11.5,-0.3) node [right]  {\ {$2$ players}}; 

   \draw (0,-5pt)   node [below] {$v$};
   \draw (10,-5pt)   node [below] {$w$};
         
   \draw node(c) [double] at (8/10,0) {} ;
   \draw (4/10,0)   node [above, font=\scriptsize] {$\xi$};   

   \draw node(d) [single] at (24/10,0) {} ;         
   \draw (16/10,0)   node [above, font=\scriptsize] {$2\xi$};

   \draw node(d) [single] at (36/10,0) {} ;
   \draw (30/10,0)   node [above, font=\scriptsize] {$\alpha(e)\xi$}; 
  
   \draw (42/10,0)   node [above, font=\scriptsize] {$\cdots$}; 
  
   \draw node(f) [single] at (48/10,0) {} ;
   \draw (54/10,0)   node [above, font=\scriptsize] {$\alpha(e)\xi$};  

   \draw node(f) [single] at (60/10,0) {} ;
   \draw (68/10,0)   node [above, font=\scriptsize] {$2\xi$};

   \draw node(h) [double] at (76/10,0) {} ;
   \draw (84/10,0)   node [above, font=\scriptsize] {$2\xi$};   
      
   \draw node(n) [double] at (92/10,0) {} ;
   \draw (96/10,0)   node [above, font=\scriptsize] {$\xi$};   

   \end{scope}
   
\end{tikzpicture}
~\vspace{0cm} \caption{\label{fi:ELL} Players in $e\in E_{LL}$.}
\end{figure}
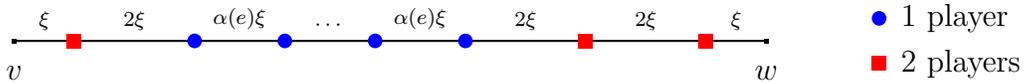
\end{center}

\bigskip

The total number of players on $S$ can be easily computed as follows. Given that each vertex $v \in V \setminus V_{L}$ has $\degree(v)$ players, there is a total of $2\card(E_{II})+\card(E_{IL})$ players on the vertices. Moreover, for each $e \in E_{II}$ there are $\lceil \lambda(e)/2\xi \rceil +1$ players in the interior of $e$; for each $e\in E_{IL}$ there are $\lceil \lambda(e)/2\xi \rceil +2$ players in the interior of $e$; for each $e\in E_{LL}$ there are $\lceil \lambda(e)/2\xi \rceil +3$ players in the interior of $e$. Hence the total number of players is
\begin{align*}
&2 \card(E_{II}) + \card(E_{IL}) + \card(E_{II}) + 2 \card(E_{IL}) + 3 \card(E_{LL}) + \sum_{e \in E}\left\lceil \frac{\lambda(e)}{2\xi} \right\rceil  \\
&\qquad= 3 \card(E) + \sum_{e \in E}\left\lceil \frac{\lambda(e)}{2\xi} \right\rceil = f(\xi) = n'.
\end{align*}

To prove that what we have constructed is a Nash equilibrium, we need to show that for all $e\in E$, we have 
\begin{equation}\label{eq:alpha}
1 \le \alpha(e) \le 2.
\end{equation}
Under $\boldsymbol{x}^{*}$, if inequality \eqref{eq:alpha} is satisfied, the payoff of each player is between $\xi$ and $2\xi$. Moreover, if a player deviated on an interval between two other players, then she would obtain a payoff equal to half the length of that interval. If inequality \eqref{eq:alpha} is satisfied, then no interval between players is longer than $2\xi$.

Furthermore, in this construction, all players who share a location with some other player have a payoff equal to $\xi$. This implies that if a player  deviates to a location that already has more than one player, then her payoff becomes less than $\xi$.
Therefore no player has a profitable deviation.
\end{proof}

\begin{claim}\label{cl:EL}
For all $e \in E_{IL}$, \eqref{eq:alpha} holds.
\end{claim}

\begin{proof}
Since $\left\lceil x \right\rceil \geq x$, we have that $2 \left\lceil \lambda(e)/(2\xi) \right\rceil -6 - \lambda(e)/\xi + 7 \geq 0$. Given \eqref{eq:alphaEL}, this implies that $\alpha(e) \le 2$. On the other hand $\alpha(e)\ge 1$ if and only if $\lambda(e)/\xi - 7 - \left\lceil \lambda(e)/(2\xi) \right\rceil +3 \geq 0$. To show that this inequality holds consider that
\[
\frac{\lambda(e)}{\xi} - \left\lceil \frac{\lambda(e)}{2\xi} \right\rceil -4 \geq \frac{\lambda(e)}{\xi} - \frac{\lambda(e)}{2\xi} -5 = \frac{\lambda(e)}{2\xi} - 5.
\]
Since, by \eqref{eq:barn},
\[
f(\xi)= n' \ge \bar{n}= f\left(\frac{\min_{e'\in E} \lambda(e')}{10}\right)
\]
and $f$ is weakly decreasing, we have 
\[
\xi \le \frac{\min_{e'\in E} \lambda(e')}{10} \le \frac{\lambda(e)}{10}
\quad\text{and therefore}\quad \frac{\lambda(e)}{2\xi}  -5 \ge 0. \qedhere
\]
\end{proof}

\begin{claim}\label{cl:EI}
For all $e \in E_{II}$, \eqref{eq:alpha} holds.
\end{claim}

\begin{proof}
Given \eqref{eq:alphaEI}, it is clear that $\alpha(e) \le 2$.
On the other hand $\alpha(e)\ge 1$ if and only if 
\[
\frac{\lambda(e)}{\xi} - 6 - \left\lceil \frac{\lambda(e)}{2\xi} \right\rceil +2 \geq 0.
\]
The left hand side is larger than $\lambda(e)/(2\xi) - 5$. As mentioned in the proof of Claim~\ref{cl:EL}, expression~\eqref{eq:barn} implies $\lambda(e)/(2\xi)  -5 \ge 0. \qedhere$
\end{proof}

\begin{claim}\label{cl:ELL}
For all $e \in E_{LL}$, \eqref{eq:alpha} holds.
\end{claim}

\begin{proof}
Given \eqref{eq:alphaELL}, it  is clear that $\alpha(e) \le 2$. On the other hand $\alpha(e)\ge 1$ if and only if 
\[
\frac{\lambda(e)}{\xi} - 8 - \left\lceil \frac{\lambda(e)}{2\xi} \right\rceil +4 \geq 0.
\]
The left hand side is larger than $\lambda(e)/(2\xi) - 5$, and again, \eqref{eq:barn} implies $\lambda(e)/(2\xi) -5 \ge 0. \qedhere$
\end{proof}

\begin{proof}[Proof of Theorem~\ref{th:existencebign}]
Proposition~\ref{pr:equilibriumnprime} shows that for every 
\begin{equation}\label{eq:boundn}
n \ge 3 \card(E) + \sum_{e\in E} \left\lceil \frac{5\lambda(e)}{\min_{e'\in E}\lambda(e')} \right\rceil
\end{equation}
there exists an integer $n \le n' \le n+\card(E)$ such that the game  $\mathcal{L}(n', S)$ admits a Nash equilibrium. Given such an equilibrium, we now construct an equilibrium for a new game with $n'-k$ players where $0 \le k \le \card(E)$. This can be achieved by removing redundant players (see Definition~\ref{def:redundant}), as follows.

We start with the equilibrium of Proposition~\ref{pr:equilibriumnprime}. For edges $e=(v,w)\in E_{IL}$ we can remove one of the two players whose distance from $v$ is $2\xi$. For edges $e=(v,w)\in E_{II}$ we can remove one of the two players whose distance from $w$ is $2\xi$. Finally for edges $e=(v,w)\in E_{LL}$ we can remove one of the two players whose distance from $w$ is $3\xi$. This way we can remove any number $k$ of players with $0\le k \le \card(E)$. We now show  that the removed players are redundant. 

In fact the above removal of players does not change the size of any interval between players, but only the payoff of  $k$ players, who now gain $2\xi$ rather than $\xi$, therefore, it does not produce any opportunity of profitable deviation for any other player, given that every one of them gains at least $\xi$. 

This proves that, for every $n$ that satisfies inequality \eqref{eq:boundn}, the game   $\mathcal{L}(n, S)$  admits an equilibrium. 
\end{proof}

\begin{remark}\label{re:Palvolgy}
Our proof differs from the one in \citet{Pal:mimeo2011} in the following respects:
\begin{enumerate}[(a)]
\item
The function $f$ defined in \eqref{eq:definitionf} takes values in $\mathbb{N}\cap[4\card(E),\infty)$ and is decreasing, but is not onto. As a consequence, there exist values $n$ such that for no  $\xi$ we have $f(\xi)=n$. 
The proof of \citet{Pal:mimeo2011} is based on a similar function, which, like ours, is in general not onto, hence his result holds only in the special case where the function $f$ is indeed onto (for instance when  the lengths of the  edges are all different).

\item
To achieve a general result, we introduce the notion of redundant players. 
This implies that players' arrangement on edges is different from the one found in \citet{Pal:mimeo2011}, in particular,  some extra players are paired in some locations. 
This in turn changes the distance between players on edges. One of these extra paired players is redundant and can therefore be removed without affecting the equilibrium of the game with $n-1$ players. Our proof shows that this argument can be repeated $\card(E)$ times, which is larger than the magnitude of the maximum jump of the function $f$. Hence, equilibria exist for every $n \ge \bar{n}$.

\item
As a consequence of our construction,  our threshold $\bar{n}$ in \eqref{eq:barn} is different from the one in \citet{Pal:mimeo2011}.

\end{enumerate}

\end{remark}

\section{Efficiency of equilibria}\label{se:efficiency}

A location game $\mathcal{L}(n, S)$ is a constant-sum game. Therefore, since any strategy profile produces the same total payoff for the retailers, it is efficient. Hence, to obtain a meaningful result, we measure the efficiency of equilibria in terms of the consumers' traveling cost and not in terms of the players' payoffs, as it is usually the case.

Consider a game $\mathcal{L}(n, S)$.
For $\boldsymbol{x}=(x_{1}, \dots, x_{n}) \in S^{n}$ and $y \in S$ define 
\[
d(\boldsymbol{x}, y) := \min_{i\in\{1,\dots,n\}}d(x_{i},y).
\]
This is the distance between a consumer located in $y$ and the closest retailer when the strategy profile $\boldsymbol{x}$ is played. 
The \emph{social cost} $C(\boldsymbol{x})$ is defined as
\[
C(\boldsymbol{x}):= \int_{S} d(\boldsymbol{x},y) \diff \lambda(y).
\]
This is the total cost incurred by the consumers, when each one of them shops at the closest store. 

\begin{definition}
Consider a game $\mathcal{L}(n, S)$ that admits a Nash equilibrium. We denote $\mathcal{E}_{n}$ the set of pure Nash equilibria of the  game $\mathcal{L}(n, S)$ and  define
\begin{enumerate}

\item
the \emph{price of anarchy}
\[
\IPoA(n) := \frac{\sup_{\boldsymbol{x} \in \mathcal{E}_{n}} C(\boldsymbol{x})}{\inf_{\boldsymbol{x} \in S^n} C(\boldsymbol{x})},
\]

\item
the \emph{price of stability}
\[
\IPoS(n) := \frac{\inf_{\boldsymbol{x} \in \mathcal{E}_{n}} C(\boldsymbol{x})}{\inf_{\boldsymbol{x} \in S^n} C(\boldsymbol{x})}.
\]

\end{enumerate}
\end{definition}

Since there always exists a positive mass of consumers at a strictly positive distance from the closest possible retailer, we have that $\inf_{\boldsymbol{x} \in S^n} C(\boldsymbol{x}) > 0$, therefore both $\IPoA(n)$ and $\IPoS(n)$ are well defined.

The next theorem shows that asymptotically the  price of anarchy cannot exceed $2$. As proved in the following sections, the result holds exactly and not only asymptotically for simple configurations of the network, but not in general. 
The same theorem shows an asymptotic result on the price of stability.

\begin{theorem}\label{th:ipoageneral}
Consider the sequence of games  $\mathcal{L}(n, S)$.
Then  
\begin{enumerate}[{\rm (a)}]
\item\label{it:th:ipoageneral-a}
there exists a function $\Phi:\mathbb{N}\to\mathbb{R}$ such that, whenever $n$ is large enough for $\mathcal{L}(n, S)$ to admit a Nash equilibrium, we have
\[
\PoA (n) \le \Phi(n) \quad\text{and}\quad \lim_{n \to \infty} \Phi(n) = 2.
\]

\item\label{it:th:ipoageneral-b}
\[ 
\lim_{n \to \infty} \PoS (n) = 1.
\]
	
\end{enumerate}
\end{theorem}

The interpretation of Theorem~\ref{th:ipoageneral} is that, when the number of retailers is large, if they are left to their own devices and play a bad Nash equilibrium, the outcome of their actions could decrease efficiency by a factor of two, approximately. On the other hand, if a planner cajoles them into playing a suitable Nash equilibrium, then efficiency is almost achieved. 

Although $\IPoS(n) \leq 2$, when $S$ is the unit interval or the circle, for any $n$ for which the equilibrium exists, this property is not true in the general case: an example of a location game on a star (see Remark \ref{re:star}) shows that the price of anarchy can be larger than $2$, hence the bound of Theorem~\ref{th:ipoageneral}\ref{it:th:ipoageneral-a} holds only asymptotically.

\begin{remark} \label{re:entrycost}
Note that in our model neither entry costs nor production costs are present.
To accommodate entry costs, we should consider a different model with an outside option, that is, the players' action space should be $S \cup \{\OUT\}$, where $\OUT$ means that a player does not enter the market. Since the total mass of consumers is fixed, a large number of players would imply a payoff smaller than the entry cost for some of them, and therefore these players would choose the action $\OUT$.
As far as production cost are concerned, since the price $p$ is exogenous, our model implicitly assumes that the fixed production cost is zero and the marginal production cost is smaller than $p$, so that the payoff of each retailer is increasing in her market share.
\end{remark}

\begin{remark}\label{re:vetta}
\citet{Vet:FOCS2002} studies a class of games, called \emph{valid utility games}, where players choose facility locations and he proves that the price of anarchy for this class of games is bounded above by $2$. Remark~\ref{re:star} below shows that, despite some similarities, locations games, as studied in our paper, are not valid utility games.

\citet{Rou:JACM2015} introduces the class of $(\lambda, \mu)$-smooth games and uses it to prove bounds for the price of anarchy of games in this class.
For instance valid utility games are $(1,1)$-smooth and Vetta's bound on the price of anarchy can be easily proved with smoothness tools. 
Unfortunately smoothness techniques do not seem to be useful for location games. The reason being that these games are payoff-maximization games with a finite number of players, but the objective function that is used to compute the price of anarchy is a cost function, and it measures the cost incurred by the continuum of consumers. In particular it is not possible to find any useful inequality between the retailers' payoff and the consumers' cost.
Moreover, a full theory of the use of smoothness to bound the price of stability has not been developed, yet. 

\end{remark}

\subsection{Proofs}

We introduce some concepts in the theory of majorization that will be used to prove some results about efficiency of equilibria. 
We refer the reader to \citet{MarOlkArn:Springer2011} for an extensive analysis of this topic. 

\begin{definition}
Given a vector $\boldsymbol{z}=(z_{1}, \dots, z_{n})$, call $z_{[1]} \ge \dots \ge z_{[n]}$ its decreasing rearrangement.
Let $\boldsymbol{x}, \boldsymbol{y} \in \mathbb{R}^{n}_{+}$ be such that 
\[
\sum_{i=1}^{n} x_{i} = \sum_{i=1}^{n} y_{i}
\]
and, for all $k \in \{1, \dots, n\}$
\[
\sum_{i=1}^{k} x_{[i]} \le \sum_{i=1}^{k} y_{[i]}.
\]
Then we say that $\boldsymbol{x}$ is \emph{majorized} by $\boldsymbol{y}$ ($\boldsymbol{x} \prec \boldsymbol{y}$).
\end{definition}

\begin{definition}
A function $\phi:\mathbb{R}^n_{+} \rightarrow \mathbb{R}$ is said to be \emph{Schur-convex} if $\boldsymbol{x} \prec \boldsymbol{y}$ implies $\phi(\boldsymbol{x}) \le \phi(\boldsymbol{y})$.
\end{definition}

\begin{lemma}\label{le:sum_convex}
If $\psi: \mathbb{R}_{+} \to \mathbb{R}$ is a convex function and 
\[
\phi(x_{1}, \dots, x_{n}) = \sum_{i=1}^{n} \psi(x_{i}),
\]
then $\phi$ is Schur-convex.
\end{lemma}

\begin{definition}\label{de:halfinterval}
Let $\boldsymbol{x}$ satisfy the vertex property as in Definition~\ref{de:VP}. Then, for $a,b\in S$, we call $[a,b]$  an \emph{$\boldsymbol{x}$-half interval} if either
\begin{enumerate}[(i)]
\item\label{it:de:halfinterval-1}
there exist $e\in E$ and  $i \in N$ such that $b\in e$ is a leaf, $x_{i}=a\in e$, and for no $j\in N$ we have $x_{j} \in (a,b]$, or

\item\label{it:de:halfinterval-2}
there exist $e \in E$ and $i,\ell \in N$ such that $x_{i}=a \in e$, $x_{\ell}\in e$, for no $j \in N$  we have $x_{j}$ between $a$ and $x_{\ell}$, and $d(a,b)=d(b,x_{\ell})=d(a,x_{\ell})/2$, i.e., $b$ is the middle point between $a$ and $x_{\ell}$.
\end{enumerate}
In both cases the roles of $a$ and $b$ can be interchanged.
\end{definition}

Basically, when a profile satisfies the vertex property, a half interval indicates the share of consumers that retailers in a location $x_{i}$ attract along one direction emanating from $x_{i}$. This could be either the whole interval from $x_{i}$ to a leaf (condition~\ref{it:de:halfinterval-1}) or the interval from $x_{i}$ to the midpoint between $x_{i}$ and $x_{\ell}$ (condition~\ref{it:de:halfinterval-2}).

In profile $\boldsymbol{x}$, if $m$ players share the same location, then we use the convention that there are $2(m-1)$ zero-length $\boldsymbol{x}$-half intervals between them. If profile $\boldsymbol{x}$ satisfies the vertex property then the whole graph can be covered with $\boldsymbol{x}$-half intervals.
We call $H(\boldsymbol{x})$ the class of all $\boldsymbol{x}$-half intervals in $S$.
We denote $\Lambda:=\lambda(S)$.

\begin{lemma}\label{le:Lovern}
Given a strategy profile $\boldsymbol{x}$, there exists $i \in  N$ such that $\rho_{i}(\boldsymbol{x}) \leq \Lambda/n$.  
\end{lemma}

\begin{proof}
If $\rho_{i}(\boldsymbol{x}) > \Lambda/n$ for all $i \in N$, then $\sum_{i=1}^{n} \rho_{i}(\boldsymbol{x}) > \Lambda$, which is a contradiction, since $\sum_{i=1}^{n} \rho_{i}(\boldsymbol{x}) = \Lambda$. 
\end{proof}

\begin{lemma}\label{le:2Lovern}
If $\boldsymbol{x}^{*}$ is a Nash equilibrium of  $\mathcal{L}(n, S)$, then for all $y \in S$ we have $d(\boldsymbol{x}^{*},y) \leq 2\Lambda/n$. 
\end{lemma}

\begin{proof}
Suppose, by contradiction that there exists $y_{0}\in S$ is such that $d(\boldsymbol{x}^{*},y_{0}) > 2\Lambda/n$. By Lemma~\ref{le:Lovern} there exists a player whose payoff is less than or equal to $\Lambda/n$. This player could deviate to $y_{0}$ and then attract at least half the consumers between $y_{0}$ and the closest player, namely she could get a payoff larger than $\Lambda/n$, making the deviation profitable.
\end{proof}

\begin{lemma}\label{le:largen}
Let $n > \bar{n}$, with $\bar{n}$ defined as in \eqref{eq:nbardef}. Assume that 
$\boldsymbol{x}^{*}$ is a Nash equilibrium of $\mathcal{L}(n, S)$. Then $\boldsymbol{x}^{*}$ satisfies the vertex property.
\end{lemma}

\begin{proof}
We have 
\begin{equation*}
\bar{n}= 3 \card(E) + \sum_{e \in E} \left\lceil \frac{5\lambda(e)}{\min_{e'\in E}\lambda(e')} \right\rceil
\ge
 3 \card(E) +  \frac{5\Lambda}{\min_{e'\in E}\lambda(e')}> \frac{4\Lambda}{\min_{e'\in E}\lambda(e')}
\end{equation*}
If $n > \bar{n}$, then for any edge $e\in E$ we have $\lambda(e) \ge 4\Lambda/n$ and, therefore, by Lemma~\ref{le:Lovern}, there are at least two players on $e$.

Take $v_{0}\in V_{I}$ and assume, \emph{ad absurdum}, that no player is in $v_{0}$. Let $i$ be the player whose location $x^{*}_{i}$ is the closest to $v_{0}$ ($i$ is not necessarily unique). If player $i$ moves towards $v_{0}$ by $\varepsilon < d(x^{*}_{i},v_{0})$, then she loses $\varepsilon/2$ on the edge where she resides, but she gains $(\degree(v_{0})-1)\varepsilon/2$ on the other incident edges on $v_{0}$. Therefore moving towards $v_{0}$ is a profitable deviation, which contradicts the assumption that $\boldsymbol{x}^{*}$ is an equilibrium.
\end{proof}

\begin{lemma}\label{le:numberhalfintervals}
Let $\boldsymbol{x}$ satisfy the vertex property. Then the number of $\boldsymbol{x}$-half intervals in $S$ is $2n+2\card(E)-\card(V_{I})-\card(V)$.
\end{lemma}

\begin{proof}
Placing one player on each $v\in V_{I}$ creates $\card(E)$ intervals between two vertices. Every time a new player is placed on some edge, a new interval is created (by splitting an existing interval into two). This is true also if the new player is placed in the same location of an existing player, since this creates two zero-length half intervals. Therefore, once all $n$ players are placed on $S$, there are exactly $\card(E)+n-\card(V_{I})$ intervals. Each on them contains two $\boldsymbol{x}$-half intervals, except the ones between a player and a leaf, which contain one half interval. Therefore the number of  $\boldsymbol{x}$-half intervals is  $2\card(E)+2n-\card(V_{I})-\card(V)$.
\end{proof}

\begin{lemma}\label{le:equilibriumhalfinterval}
Assume that $\boldsymbol{x}^{*}$ is an equilibrium  of $\mathcal{L}(n, S)$ and $[a,b]$ is an $\boldsymbol{x}^{*}$-half interval. Then $\lambda([a,b]) \le \Lambda/n$.
\end{lemma}

\begin{proof}
Assume, \emph{ad absurdum}, that $\lambda([a,b]) > \Lambda/n$. By Lemma~\ref{le:Lovern}, there exists $i \in  N$ such that $\rho_{i}(\boldsymbol{x}^{*}) \leq \Lambda/n$. Two cases are possible.

\noindent Case~\ref{it:de:halfinterval-1} of Definition~\ref{de:halfinterval}. If player $i$ deviates to $[a,b]$ at a distance $\varepsilon$ from $a$, then, for $\varepsilon$ small enough, her payoff becomes $\lambda([a,b])-\varepsilon > \Lambda/n$.  
  
\noindent Case~\ref{it:de:halfinterval-2} of Definition~\ref{de:halfinterval}. 
If player $i$ deviates to $b$, then her payoff becomes $\lambda([a,b]) > \Lambda/n$.

The existence of profitable deviations contradicts the assumption that $\boldsymbol{x}^{*}$ is a Nash equilibrium.
\end{proof}

\begin{lemma}\label{le:boundeqsocialcost}
Assume that the conditions of Lemma~\ref{le:largen} are satisfied.
Then $C(\boldsymbol{x}^{*}) \le \Lambda^{2}/2n$.
\end{lemma}

\begin{proof}
By Lemma~\ref{le:largen}, $\boldsymbol{x}^{*}$ satisfies the vertex property. From the definition of social cost $C$ it follows that
\[
C(\boldsymbol{x}^{*}) = \sum_{[a,b]\in H(\boldsymbol{x}^{*})}\frac{\lambda([a,b])^{2}}{2}.
\]
Call $\boldsymbol{\lambda}(\boldsymbol{x}^{*})$ the vector of all $\lambda([a,b])$ such that $[a,b]\in H(\boldsymbol{x}^{*})$. By Lemma~\ref{le:equilibriumhalfinterval}, $\boldsymbol{\lambda}(\boldsymbol{x}^{*})$ is dominated in the majorization order by the vector
$(\Lambda/n, \dots, \Lambda/n, 0, \dots, 0)$,
where the number of positive components is $n$. Since the function $(z_{1}, \dots, z_{m})\mapsto \sum_{i=1}^{m}z_{1}^{2}/2$ is Schur-convex, we have 
\[
C(\boldsymbol{x}^{*}) \le \sum_{i=1}^{n} \frac{1}{2}\left(\frac{\Lambda}{n}\right)^{2}=\frac{\Lambda^{2}}{2n}. \qedhere
\]
\end{proof}

\begin{lemma}\label{le:boundsocialoptimum}
The following inequality holds:
\begin{equation*}
\inf_{\boldsymbol{x} \in S^n} C(\boldsymbol{x}) \geq     \frac{\Lambda^{2}}{2 (2n+2\card(E_{II})+\card(E_{IL}))}.
\end{equation*}
\end{lemma}

\begin{proof}

Call $\widetilde{N}=N \cup V_{I}$ a fictitious set of players obtained by adding to the original set of players $N$ one player for each vertex of degree larger than $2$ and define $\widetilde{n}=\card(\widetilde{N})$. It is clear that 
\begin{equation*}
\inf_{x \in S^n} C(x) \ge C(\widetilde{\boldsymbol{x}}),
\end{equation*}
where the profile $\widetilde{x}$ contains $\widetilde{n}$ players: $n$ players are located according to social optimum, and $\widetilde{n}-n$ players located on each vertex of degree larger than $3$, unoccupied in the optimum.

Applying the argument used in Lemma~\ref{le:numberhalfintervals} to this new profile, we can show that the number of $\widetilde{\boldsymbol{x}}$-half intervals is 
\begin{align*}
M:=&2n+2\card(E)-\card(V)+\card(V_{I})\\
=& 2n+2\card(E)-\card(V_{L})\\
=& 2n+2\card(E)-\card(E_{IL})-2\card(E_{LL})\\
=& 2n+2\card(E_{II})+\card(E_{IL}).
\end{align*}
It is clear that $\boldsymbol{\lambda}(\widetilde{\boldsymbol{x}})$ dominates the vector 
$(\Lambda/M, \dots, \Lambda/M)$.
Since 
\[
\sum_{i=1}^{M}\frac{1}{2}\left(\frac{\Lambda}{M}\right)^{2}=\frac{\Lambda^{2}}{2M},
\]
we have
\[
\inf_{x \in S^n} C(x) \geq C(\widetilde{\boldsymbol{x}}) \geq \frac{\Lambda^{2}}{2 M}. \qedhere
\]
\end{proof} 

\begin{proof}[Proof of Theorem~\ref{th:ipoageneral}\ref{it:th:ipoageneral-a}] From the bounds in Lemmata~\ref{le:boundeqsocialcost} and \ref{le:boundsocialoptimum} we conclude that
\begin{equation*}
\PoA(n) \le \Phi(n):=\frac{4n+4\card(E_{II})+2\card(E_{IL})}{2n} \xrightarrow[n\to\infty]{} 2. \qedhere
\end{equation*}
\end{proof}

\begin{claim}\label{cl:socialcost}
Consider the equilibrium $\boldsymbol{x}^{*}$ constructed in the proof of Proposition~\ref{pr:equilibriumnprime}. Then 
\begin{multline*}
C(\boldsymbol{x}^{*}) = \sum_{e \in E_{IL}} \left(7 \frac{\xi^{2}}{2} + \left(\left\lceil \frac{\lambda(e)}{2 \xi}\right\rceil -3\right) \frac{\alpha(e)^{2} \xi ^{2}}{4}\right) \\
+ \sum_{e \in E_{II}} \left(6 \frac{\xi^{2}}{2} + \left(\left\lceil \frac{\lambda(e)}{2 \xi}\right\rceil -2 \right) \frac{\alpha(e)^{2} \xi ^{2}}{4}\right)
+ \sum_{e \in E_{LL}} \left(8 \frac{\xi^{2}}{2} + \left(\left\lceil \frac{\lambda(e)}{2 \xi}\right\rceil -4\right) \frac{\alpha(e)^{2} \xi ^{2}}{4}\right).
\end{multline*}
\end{claim}

\begin{proof}
Each edge $e \in E_{IL}$ contains $7$ half intervals of length $\xi$ and $\left\lceil \lambda(e)(2 \xi)\right\rceil -3$ intervals of length $\alpha(e) \xi$. The cost of edge $e$ is then
\[ 
7 \frac{\xi^{2}}{2} + \left(\left\lceil \frac{\lambda(e)}{2 \xi}\right\rceil -3\right) \frac{\alpha(e)^{2} \xi ^{2}}{4}.
\]

Each edge $e \in E_{II}$ contains $6$ half intervals of length $\xi$ and $\left\lceil \lambda(e)/(2 \xi)\right\rceil -2$ intervals of length $\alpha(e) \xi$. The cost of edge $e$ is then
\[ 
6 \frac{\xi^{2}}{2} + \left(\left\lceil \frac{\lambda(e)}{2 \xi}\right\rceil -2 \right) \frac{\alpha(e)^{2} \xi ^{2}}{4}.
\]

Each edge $e \in E_{LL}$ contains $8$ half intervals of length $\xi$ and $\left\lceil \lambda(e)/(2 \xi)\right\rceil -4$ intervals of length $\alpha(e) \xi$. The cost of edge $e$ is then
\[
8 \frac{\xi^{2}}{2} + \left(\left\lceil \frac{\lambda(e)}{2 \xi}\right\rceil -4\right) \frac{\alpha(e)^{2} \xi ^{2}}{4}. \qedhere
\]
\end{proof}

\begin{claim}\label{cl:boundxi}
\[
\frac{\Lambda}{2n-4\card(E)} \leq \xi \leq \frac{\Lambda}{2n-6\card(E)}.
\]
\end{claim}

\begin{proof}
By definition, $\xi$ is such that $f(\xi)=n'$ with $n \leq n' \leq n + \card(E)$, where $f$ is defined as in \eqref{eq:definitionf}.
Therefore
\begin{equation*}
n \leq 3 \card(E) + \sum_{e \in E} \left\lceil \frac{\lambda(e)}{2\xi} \right\rceil \leq n + \card(E),
\end{equation*}
which implies
\[             
\frac{\Lambda}{2n-4 \card(E)} \leq \xi \leq \frac{\Lambda}{2n-6 \card(E)}.  \qedhere
\]
\end{proof}

Given two functions $g$ and $h$, we say that $g(n) \underset{n \rightarrow \infty}{\sim} h(n)$ if $g(n)/h(n)\to 1$ as $n$ goes to infinity.

\begin{claim}\label{cl:alpha2}
For all $e \in E$ we have $\lim_{n\to\infty}\alpha(e) = 2$.
\end{claim}

\begin{proof}
Claim ~\ref{cl:boundxi} implies
\begin{equation}\label{eq:xin}
\xi \underset{n \rightarrow \infty}{\sim} \frac{\Lambda}{2n}.
\end{equation}
Therefore, if $e \in E_{IL}$, then
\[
\alpha(e) = \frac{\lambda(e)- 7 \xi}{\xi \left\lceil \frac{\lambda(e)}{2 \xi} \right\rceil - 3 \xi}  \underset{n \rightarrow \infty}{\sim}  2;
\]
if $e \in E_{II}$, then
\[
\alpha(e) = \frac{\lambda(e)- 6 \xi}{\xi \left\lceil \frac{\lambda(e)}{2 \xi} \right\rceil - 2 \xi}  \underset{n \rightarrow \infty}{\sim}  2;
\]
if $e \in E_{LL}$, then
\[
\alpha(e) = \frac{\lambda(e)- 8 \xi}{\xi \left\lceil \frac{\lambda(e)}{2 \xi} \right\rceil - 4 \xi}  \underset{n \rightarrow \infty}{\sim}  2. \qedhere
\]
\end{proof}

\begin{claim}\label{cl:sigman}
\[   
C(\boldsymbol{x}^{*})  \underset{n \rightarrow \infty}{\sim}  \frac{\Lambda^{2}}{4n}.  
\]
\end{claim}

\begin{proof}
Using Claims~\ref{cl:socialcost} and \ref{cl:alpha2}, we  have
\[
C(\boldsymbol{x}^{*})  \underset{\xi \to 0}{\sim}  \sum_{e \in E_{IL}} \left(\frac{\xi^{2}}{2} + \frac{\lambda(e)\xi}{2}\right)
+ \sum_{e \in E_{II}} \left(\xi^{2} + \frac{\lambda(e)\xi}{2}\right)
+ \sum_{e \in E_{LL}} \left(2\xi^{2} + \frac{\lambda(e)\xi}{2}\right).
\]
Hence
\[
C(\boldsymbol{x}^{*}) \underset{\xi \to 0}{\sim} \sum_{e \in E} \frac{\lambda(e)\xi}{2},
\]
that is, thanks to \eqref{eq:xin},
\[
C(\boldsymbol{x}^{*}) \underset{n \to \infty}{\sim} \frac{\Lambda^{2}}{4n}. \qedhere
\]
\end{proof}

\begin{proof}[Proof of Theorem~\ref{th:ipoageneral}\ref{it:th:ipoageneral-b}]
By Lemma~\ref{le:boundsocialoptimum}
we have
\[
\inf_{x \in S^n} C(x) \geq \frac{\Lambda^{2}}{4n+4\card(E_{II})+2\card(E_{IL})}.
\]
Therefore
\[
\IPoS(n) \le 
C(\boldsymbol{x}^{*})\left(\displaystyle{\frac{\Lambda^{2}}{4n+4\card(E_{II})+2\card(E_{IL})}}\right)^{-1}.
\]
Since
\[
\frac{\Lambda^{2}}{4n+4\card(E_{II})+2\card(E_{IL}))} \underset{n \to \infty}{\sim} \frac{\Lambda^{2}}{4n}
\]
and $\IPoS(n) \ge 1$, using Claim~\ref{cl:sigman},  we obtain
\[
\lim_{n \to \infty} \IPoS(n) = 1. \qedhere
\]
\end{proof}

\section{Examples}\label{se:examples}

In this section we consider some simple examples of networks and show that in some cases exact results can be obtained. 

\subsection{The circle}\label{suse:circle}

We now assume consumers to be distributed on the unit circle $\mathcal{C}$. 
This model has been studied by \citet{EatLip:RES1975}, who deal with existence of equilibria for the model without price, and by \citet{Sal:BJE1979}, who considers the model with price. 
Notice that \emph{stricto sensu} this  is not a particular case of our general model, since the circle is not a graph. We can see it as a graph where all points have degree $2$.

\begin{proposition}\label{pr:existenceC}
For every $n \ge 1$ the set of equilibria of the game $\mathcal{L}(n, \mathcal{C})$ is  non-empty. 
\end{proposition}

\begin{proposition}\label{pr:pricesalop}

In the game  $\mathcal{L}(n, \mathcal{C})$, we have:
\begin{enumerate}[{\rm (a)}]
\item\label{it:pr:ipoacircle}
\[
\IPoA(n) = 
\begin{cases}
2 & \text{if $n$ is even}, \\
\displaystyle{2\frac{n}{n+1}}& \text{if $n$ is odd}.
\end{cases}
\]

\item\label{it:pr:iposcircle}
\[
\IPoS(n) = 1.
\]
\end{enumerate}

\end{proposition}

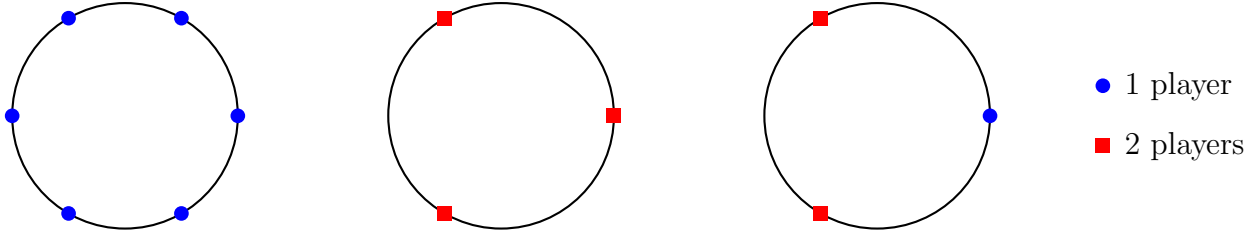
\begin{figure}[H]

\tikzstyle{leaf}=[rectangle,fill,scale=0.2]
\tikzstyle{single}=[circle,fill=blue,scale=0.5]
\tikzstyle{double}=[rectangle,fill=red,scale=0.7]
\tikzstyle{rple}=[diamond,fill=green,scale=0.6]
\tikzstyle{kple}=[regular polygon,regular polygon sides=3,fill=brown,scale=0.5]

\newdimen\mydim
\begin{tikzpicture}[auto, thick]

\mydim=1.5cm

    \draw node [single] at (13, 0.4) {};
   \draw (13,0.4) node [right]  {\ {$1$ player}};
   \draw node [double] at (13, -0.4) {};
   \draw (13, -0.4) node [right]  {\ {$2$ players}};

   \draw (0,0) circle (\mydim);
   \draw node(a) [single] at (0:\mydim) {};
   \draw node(b) [single] at (60:\mydim) {};
   \draw node(c) [single] at (120:\mydim) {};
   \draw node(d) [single] at (180:\mydim) {};
   \draw node(e) [single] at (240:\mydim) {};
   \draw node(f) [single] at (300:\mydim) {};   
      
\begin{scope}[xshift=5 cm]
 \draw (0,0) circle (\mydim);
   \draw node(a) [double] at (0:\mydim) {};
   \draw node(c) [double] at (120:\mydim) {};
   \draw node(e) [double] at (240:\mydim) {};
\end{scope}

 \begin{scope}[xshift=10 cm]
   \draw (0,0) circle (\mydim);
   \draw node(a) [single] at (0:\mydim) {};
   \draw node(c) [double] at (120:\mydim) {};
   \draw node(e) [double] at (240:\mydim) {};
\end{scope}
\end{tikzpicture}
~\vspace{0cm} \caption{\label{fi:bestC} Left: best equilibrium $\widetilde{\boldsymbol{x}}$ with $6$ players; middle: worst equilibrium $\widehat{\boldsymbol{x}}$ with $6$ players; right: worst equilibrium $\breve{\boldsymbol{x}}$  with $5$ players, on $\mathcal{C}$.}
\end{figure}

\subsection{The segment}\label{suse:segment}

The model described in this subsection was studied in details by  \citet{EatLip:RES1975} under slightly different assumptions.
We consider the location game on a segment, which, without loss of generality, is assumed to be $[0,1]$.

\begin{proposition}\label{pr:zeroone}
Consider the location game $\mathcal{L}(n,[0,1])$.
\begin{enumerate}[{\rm (a)}]
\item\label{it:pr:zeroone-a}
For $n=2,4,5$ there exists a unique (modulo permutation of players) pure Nash equilibrium 

\item\label{it:pr:zeroone-b}
For $n=3$, there is no pure Nash equilibrium.

\item\label{it:pr:zeroone-c} 
For $n \ge 6$, there is an infinite number of pure Nash equilibria. 

\end{enumerate}
\end{proposition}

If, without any loss of generality, we assume that the equilibrium $\boldsymbol{x}^{*}$ satisfies  $x^{*}_{i} \le x^{*}_{i+1}$ and we call $\eta_{i} = x^{*}_{i+4} - x^{*}_{i+3}$, then Figure~\ref{fi:nplayers} is an example of $n$-player Nash equilibrium if and only if
\begin{enumerate}[(i)]
\item
for all $i \in \{1,\dots,n-5\}$, we have  $\eta_{i} \geq \xi$,
\item
for all $i \in \{1,\dots,n-6\}$, we have  $(\eta_{i}+\eta_{i+1})/2 \leq \xi$.

\end{enumerate}

\bigskip

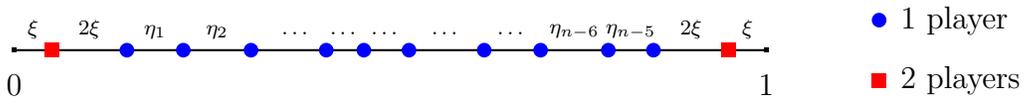
\begin{figure}[H]
\centering

\tikzstyle{leaf}=[rectangle,fill,scale=0.2]
\tikzstyle{single}=[circle,fill=blue,scale=0.5]
\tikzstyle{double}=[rectangle,fill=red,scale=0.7]
\tikzstyle{rple}=[diamond,fill=green,scale=0.6]
\tikzstyle{kple}=[regular polygon,regular polygon sides=3,fill=brown,scale=0.5]

\begin{tikzpicture}

  \begin{scope}[thick]
  
   \draw node(l0) [leaf] at (0,0) {};
   \draw node(l1) [leaf] at (10,0) {};
   
   \draw  (l0) -- (l1);
          
   \draw node(a) [single] at (11.5, 0.4) {};
   \draw (11.5,0.4) node [right]  {\ {$1$ player}}; 
       
   \draw node(b) [double] at (11.5, -0.4) {};
   \draw (11.5,-0.4) node [right]  {\ {$2$ players}};

   \draw (0,-5pt)   node [below] {$0$};
   \draw (10,-5pt)   node [below] {$1$};
         
   \draw node(c) [double] at (10/20,0) {} ;
   \draw (5/20,0)   node [above, font=\scriptsize] {$\xi$};

   \draw node(d) [single] at (30/20,0) {} ;
   \draw (20/20,0)   node [above, font=\scriptsize] {$2\xi$};           

   \draw node(e) [single] at (45/20,0) {} ;
   \draw (37.5/20,0)   node [above, font=\scriptsize] {$\eta_{1}$};  
   
   \draw node(f) [single] at (63/20,0) {} ;
   \draw (54/20,0)   node [above, font=\scriptsize] {$\eta_{2}$};  
   
   \draw (75.2/20,0)   node [above, font=\scriptsize] {$\cdots$};     
     
   \draw node(g) [single] at (83/20,0) {} ;
   \draw (88/20,0)   node [above, font=\scriptsize] {$\cdots$};     

   \draw node(h) [single] at (93/20,0) {} ;
   \draw (99/20,0)   node [above, font=\scriptsize] {$\cdots$};  

   \draw node(i) [single] at (105/20,0) {} ;
   \draw (115/20,0)   node [above, font=\scriptsize] {$\cdots$};  
      
   \draw node(j) [single] at (125/20,0) {} ;
   \draw (132.5/20,0)   node [above, font=\scriptsize] {$\cdots$};  
      
   \draw node(k) [single] at (140/20,0) {} ;
   \draw (149/20,0)   node [above, font=\scriptsize] {$\eta_{n-6}$};  
                              
   \draw node(l) [single] at (158/20,0) {} ;
   \draw (164/20,0)   node [above, font=\scriptsize] {$\eta_{n-5}$};             

   \draw node(m) [single] at (170/20,0) {} ;
   \draw (180/20,0)   node [above, font=\scriptsize] {$2\xi$};   
      
   \draw node(n) [double] at (190/20,0) {} ;
   \draw (195/20,0)   node [above, font=\scriptsize] {$\xi$};   
        
   \end{scope}
   
\end{tikzpicture}
~\vspace{0cm} \caption{\label{fi:nplayers} Equilibrium with $n$ players.}
\end{figure}

\begin{proposition}\label{pr:ipoaipossegment}
In the game  $\mathcal{L}(n,[0,1])$, we have:
\begin{enumerate}[{\rm (a)}]
\item\label{it:pr:ipoasegment}
\[
\IPoA(n) = 
\begin{cases}
2 & \text{if $n$ is even}, \\
\displaystyle{2\frac{n}{n+1}}& \text{if $n > 3$ is odd}.
\end{cases}
\]

\item\label{it:pr:ipossegment}
For $n=2$
\[
\IPoS(n) = 2.
\]
For $n \ge 4$
\[
\IPoS(n) = \frac{n}{n-2}.
\]
\end{enumerate}

\end{proposition}

\begin{figure}[H]
\centering

\tikzstyle{leaf}=[rectangle,fill,scale=0.2]
\tikzstyle{single}=[circle,fill=blue,scale=0.5]
\tikzstyle{double}=[rectangle,fill=red,scale=0.7]
\tikzstyle{rple}=[diamond,fill=green,scale=0.6]
\tikzstyle{kple}=[regular polygon,regular polygon sides=3,fill=brown,scale=0.5]

\begin{tikzpicture}

  \begin{scope}[scale=0.75]
  
   \draw node(l0) [leaf] at (0,0) {};
   \draw node(l1) [leaf] at (10,0) {};
   
   \draw  (l0) -- (l1);

   \draw (0,-5pt)   node [below] {$0$};
   \draw (10,-5pt)   node [below] {$1$};

   \draw (2.5/10,0)   node [above, font=\scriptsize] {$\frac{1}{2n}$};   
         
   \draw node(c) [single] at (5/10,0) {} ;
   \draw (10/10,0)   node [above, font=\scriptsize] {$\frac{1}{n}$};   
         
   \draw node(c) [single] at (15/10,0) {} ;
   \draw (20/10,0)   node [above, font=\scriptsize] {$\frac{1}{n}$};

   \draw node(d) [single] at (25/10,0) {} ;

   \draw node(e) [single] at (35/10,0) {} ;
   
   \draw node(f) [single] at (45/10,0) {} ;
   \draw (50/10,0)   node [above, font=\scriptsize] {$\cdots$};  
   
   \draw node(g) [single] at (55/10,0) {} ;

   \draw node(h) [single] at (65/10,0) {} ;

   \draw node(i) [single] at (75/10,0) {} ;
   \draw (80/10,0)   node [above, font=\scriptsize] {$\frac{1}{n}$};   

   \draw node(i) [single] at (85/10,0) {} ;   
   \draw (90/10,0)   node [above, font=\scriptsize] {$\frac{1}{n}$};   

   \draw node(i) [single] at (95/10,0) {} ;         
   \draw (97.5/10,0)   node [above, font=\scriptsize] {$\frac{1}{2n}$};   
 
   \end{scope}
   
     \begin{scope}[xshift=8.4cm,scale=0.75]
  
   \draw node(l0) [leaf] at (0,0) {};
   \draw node(l1) [leaf] at (10,0) {};
   
   \draw  (l0) -- (l1);      
       \draw (0,-5pt)   node [below] {$0$};
   \draw (10,-5pt)   node [below] {$1$};

   \draw (5/16,0)   node [above, font=\scriptsize] {$\frac{1}{2n-4}$};           
   \draw node(c) [double] at (10/16,0) {} ;
   \draw (20/16,0)   node [above, font=\scriptsize] {$\frac{1}{n-2}$};
         
   \draw node(c) [single] at (30/16,0) {} ;
   \draw (40/16,0)   node [above, font=\scriptsize] {$\frac{1}{n-2}$};           

   \draw node(e) [single] at (50/16,0) {} ;
   
   \draw node(f) [single] at (70/16,0) {} ;
   \draw (80/16,0)   node [above, font=\scriptsize] {$\cdots$};  
   
   \draw node(g) [single] at (90/16,0) {} ;

   \draw node(h) [single] at (110/16,0) {} ;

   \draw (120/16,0)   node [above, font=\scriptsize] {$\frac{1}{n-2}$};         
      
   \draw node(j) [single] at (130/16,0) {} ;

   \draw (140/16,0)   node [above, font=\scriptsize] {$\frac{1}{n-2}$};   
      
   \draw node(n) [double] at (150/16,0) {} ;
   \draw (155/16,0)   node [above, font=\scriptsize] {$\frac{1}{2n-4}$};         

   \end{scope}
   
   \begin{scope}[yshift=-2cm,scale=0.75]
  
   \draw node(l0) [leaf] at (0,0) {};
   \draw node(l1) [leaf] at (10,0) {};
   
   \draw  (l0) -- (l1);

   \draw (0,-5pt)   node [below] {$0$};
   \draw (10,-5pt)   node [below] {$1$};
         
   \draw node(c) [double] at (10/10,0) {} ;
   \draw (5/10,0)   node [above, font=\scriptsize] {$\frac{1}{n}$};   
         
   \draw (20/10,0)   node [above, font=\scriptsize] {$\frac{2}{n}$};

   \draw node(d) [double] at (30/10,0) {} ;
   \draw (40/10,0)   node [above, font=\scriptsize] {$\cdots$}; 

   \draw node(f) [double] at (50/10,0) {} ;
   \draw (60/10,0)   node [above, font=\scriptsize] {$\cdots$};  

   \draw node(h) [double] at (70/10,0) {} ;
 
   \draw (80/10,0)   node [above, font=\scriptsize] {$\frac{2}{n}$};   
      
   \draw node(n) [double] at (90/10,0) {} ;
   \draw (95/10,0)   node [above, font=\scriptsize] {$\frac{1}{n}$};   

   \end{scope}
   
   \begin{scope}[xshift=8.4cm,yshift=-2cm,scale=0.75]
  
   \draw node(l0) [leaf] at (0,0) {};
   \draw node(l1) [leaf] at (10,0) {};
   
   \draw  (l0) -- (l1);

   \draw (0,-5pt)   node [below] {$0$};
   \draw (10,-5pt)   node [below] {$1$};
         
   \draw node(c) [double] at (10/10,0) {} ;
   \draw (5/10,0)   node [above, font=\scriptsize] {$\frac{1}{n+1}$};   
         
   \draw (20/10,0)   node [above, font=\scriptsize] {$\frac{2}{n+1}$};

   \draw node(d) [double] at (30/10,0) {} ;
   \draw (40/10,0)   node [above, font=\scriptsize] {$\cdots$}; 
  
   \draw node(f) [single] at (50/10,0) {} ;
   \draw (60/10,0)   node [above, font=\scriptsize] {$\cdots$};  

   \draw node(h) [double] at (70/10,0) {} ;
  
   \draw (80/10,0)   node [above, font=\scriptsize] {$\frac{2}{n+1}$};   
      
   \draw node(n) [double] at (90/10,0) {} ;
   \draw (95/10,0)   node [above, font=\scriptsize] {$\frac{1}{n+1}$};   

   \end{scope}
   
   \begin{scope}[xshift=8.4cm,yshift=-2cm,scale=0.75]
 
   \draw node [single] at (-4, -2) {};
   \draw (-4,-2) node [right]  {\ {$1$ player}};
             
   \draw node [double] at (1, -2) {};
   \draw (1,-2) node [right]  {\ {$2$ players}};
   
   \end{scope}

\end{tikzpicture}
~\vspace{0cm} \caption{\label{fi:opteqsegment} Top left: Social optimum $\overline{\boldsymbol{x}}$ with $n$ players.
Top right: Best equilibrium $\widetilde{\boldsymbol{x}}$ with $n$ players.
Bottom left: Worst equilibrium $\widehat{\boldsymbol{x}}$ with $n$ players ($n$ even).
Bottom right: Example of worst equilibrium $\widehat{\boldsymbol{x}}^{\ell}$ with $n$ players ($n$ odd), where $\ell$ is the only unmatched  player.}
\end{figure}
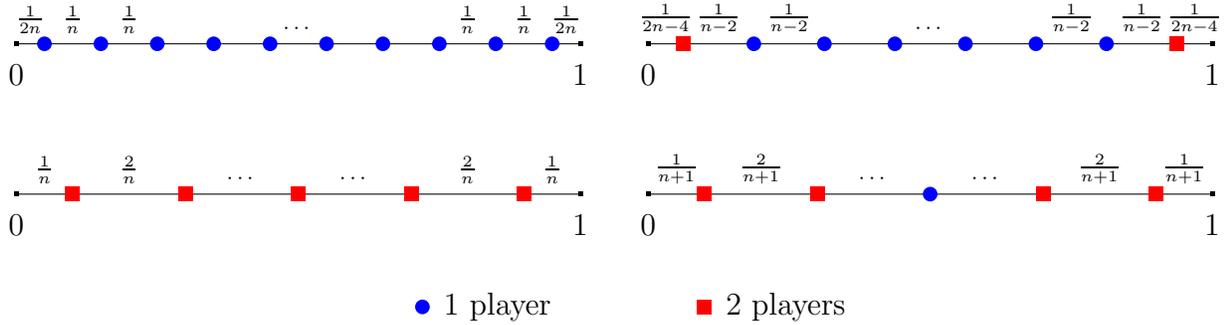

\subsection{The star}\label{suse:star}

In the whole section, we have $k>2$.
We assume $S$ to be a \emph{star} $S_{k}$, that is a network with $k+1$ vertices $\{v_{0}, v_{1}, \dots, v_{k}\}$ where for $j \in \{1, \dots, k\}$ vertex $v_{j}$ is connected to vertex $v_{0}$ and to no other vertex. 
The length of all the edges $[v_{0},v_{j}]$ is assumed to be equal to $1$.

\begin{proposition}\label{pr:barStark}
Consider a location game $\mathcal{L}(n, S_{k})$. 
\begin{enumerate}[{\rm (a)}]
\item\label{it:pr:barStark-a}
If $2 \le n \le k$, then a unique equilibrium $\boldsymbol{x}^{*}$ exists where  $x^{*}_{i} = v_{0}$ for all $i \in N$.

\item\label{it:pr:barStark-b}
If $k < n < 3k-1$, then there is no Nash equilibrium.

\item\label{it:pr:barStark-c}
If $3k-1\le{n}\le3{k}$, then there exists a unique equilibrium.

\item\label{it:pr:barStark-d}
If $3k+1\le{n}$, then there exists an infinite number of equilibria. 
\end{enumerate}
\end{proposition}

Figure~\ref{fi:eqstar} shows some examples of equilibria on the star with different numbers of players.

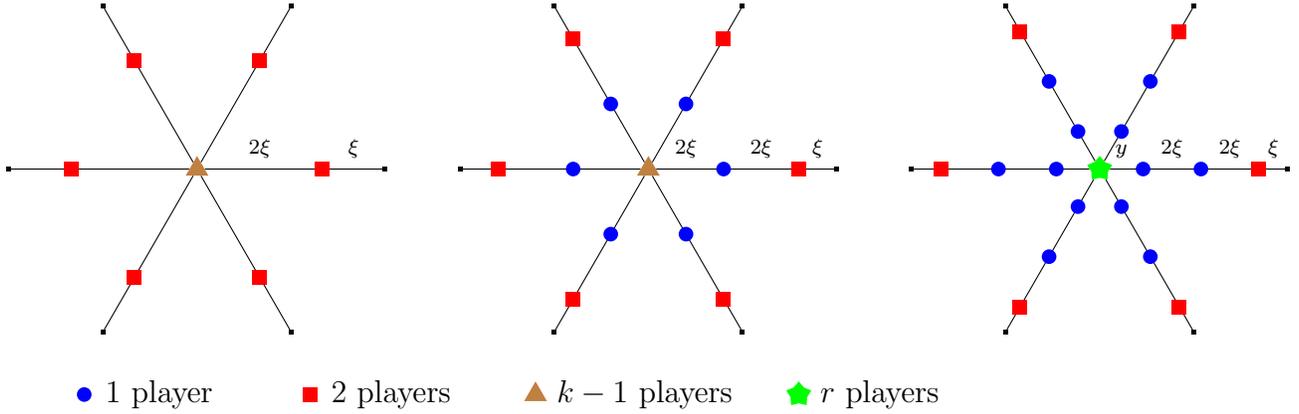
\begin{figure}[H]

\tikzstyle{leaf}=[rectangle,fill,scale=0.2]
\tikzstyle{single}=[circle,fill=blue,scale=0.5]
\tikzstyle{double}=[rectangle,fill=red,scale=0.7]
\tikzstyle{rple}=[star,fill=green,scale=0.6]
\tikzstyle{kple}=[regular polygon,regular polygon sides=3,fill=brown,scale=0.45]

\newdimen\mydim
\begin{tikzpicture}
\mydim=5cm
\begin{scope}[scale=0.5]
 
   \draw node [single] at (-3, -6) {};
   \draw (-3,-6) node [right]  {\ {$1$ player}};
             
   \draw node [double] at (3, -6) {};
   \draw (3,-6) node [right]  {\ {$2$ players}};
   
   \draw node [kple] at (9, -6) {};
   \draw (9,-6) node [right]  {\ {$k-1$ players}};
    
   \draw node [rple] at (16, -6) {};
   \draw (16,-6) node [right]  {\ {$r$ players}};

   \draw node(a) [leaf] at (0:\mydim) {};
   \draw node(b) [leaf] at (60:\mydim) {};
   \draw node(c) [leaf] at (120:\mydim) {};
   \draw node(d) [leaf] at (180:\mydim) {};
   \draw node(e) [leaf] at (240:\mydim) {};
   \draw node(f) [leaf] at (300:\mydim) {};
   \draw node(o) [kple] at (0,0) {} ;
   
   \draw (o) -- (a);     
   \draw (o) -- (b);    
   \draw (o) -- (c);   
   \draw (o) -- (d);     
   \draw (o) -- (e);    
   \draw (o) -- (f);    

   \draw node(aa) [double] at (0:2\mydim/3) {};
   \draw node(ba) [double] at (60:2\mydim/3) {};
   \draw node(ca) [double] at (120:2\mydim/3) {};
   \draw node(da) [double] at (180:2\mydim/3) {};
   \draw node(ea) [double] at (240:2\mydim/3) {};
   \draw node(fa) [double] at (300:2\mydim/3) {};

   \draw (0:5\mydim/6) node(axi3) [above, font=\scriptsize] {$\xi$};         
   \draw (60:5\mydim/6) node(axi3) {};         
   \draw (120:5\mydim/6) node(axi3) {};         
   \draw (180:5\mydim/6) node(axi3) {};         
   \draw (240:5\mydim/6) node(axi3) {};         
   \draw (300:5\mydim/6) node(axi3) {};         
   
   \draw (0:\mydim/3) node(axi2) [above, font=\scriptsize] {$2\xi$};         
   \draw (60:\mydim/3) node(axi2) {};           
   \draw (120:\mydim/3) node(axi2) {};           
   \draw (180:\mydim/3) node(axi2) {};        
   \draw (240:\mydim/3) node(axi2) {};           
   \draw (300:\mydim/3) node(axi2) {};  
   \end{scope}

\begin{scope}[scale=0.5,xshift=12cm]

   \draw node(a) [leaf] at (0:\mydim) {};
   \draw node(b) [leaf] at (60:\mydim) {};
   \draw node(c) [leaf] at (120:\mydim) {};
   \draw node(d) [leaf] at (180:\mydim) {};
   \draw node(e) [leaf] at (240:\mydim) {};
   \draw node(f) [leaf] at (300:\mydim) {};
   \draw node(o) [kple] at (0,0) {} ;
   
   \draw (o) -- (a);     
   \draw (o) -- (b);    
   \draw (o) -- (c);   
   \draw (o) -- (d);     
   \draw (o) -- (e);    
   \draw (o) -- (f);    

   \draw node(aa) [double] at (0:4\mydim/5) {};
   \draw node(ba) [double] at (60:4\mydim/5) {};
   \draw node(ca) [double] at (120:4\mydim/5) {};
   \draw node(da) [double] at (180:4\mydim/5) {};
   \draw node(ea) [double] at (240:4\mydim/5) {};
   \draw node(fa) [double] at (300:4\mydim/5) {};   

   \draw node(ab) [single] at (0:2\mydim/5) {};
   \draw node(bb) [single] at (60:2\mydim/5) {};
   \draw node(cb) [single] at (120:2\mydim/5) {};
   \draw node(db) [single] at (180:2\mydim/5) {};
   \draw node(eb) [single] at (240:2\mydim/5) {};
   \draw node(fb) [single] at (300:2\mydim/5) {};    
   
   \draw (0:9\mydim/10) node(axi1) [above, font=\scriptsize] {$\xi$};         
   \draw (60:9\mydim/10) node(axi1) {};           
   \draw (120:9\mydim/10) node(axi1) {};  
   \draw (180:9\mydim/10) node(axi1) {};        
   \draw (240:9\mydim/10) node(axi1) {};  
   \draw (300:9\mydim/10) node(axi1) {};  
   
   \draw (0:3\mydim/5) node(axi2) [above, font=\scriptsize] {$2\xi$};         
   \draw (60:3\mydim/5) node(axi2) {};          
   \draw (120:3\mydim/5) node(axi2) {};           
   \draw (180:3\mydim/5) node(axi2) {};          
   \draw (240:3\mydim/5) node(axi2) {};  
   \draw (300:3\mydim/5) node(axi2)  {};  
   
   \draw (0:1\mydim/5) node(axi3) [above, font=\scriptsize] {$2\xi$};         
   \draw (60:1\mydim/5) node(axi3) {};           
   \draw (120:1\mydim/5) node(axi3) {};           
   \draw (180:1\mydim/5) node(axi3) {};           
   \draw (240:1\mydim/5) node(axi3) {};           
   \draw (300:1\mydim/5) node(axi3) {};           
\end{scope}

\begin{scope}[scale=0.5,xshift=24cm]

   \draw node(a) [leaf] at (0:\mydim) {};
   \draw node(b) [leaf] at (60:\mydim) {};
   \draw node(c) [leaf] at (120:\mydim) {};
   \draw node(d) [leaf] at (180:\mydim) {};
   \draw node(e) [leaf] at (240:\mydim) {};
   \draw node(f) [leaf] at (300:\mydim) {};
   \draw node(o) [rple] at (0,0) {} ;
   
   \draw (o) -- (a);     
   \draw (o) -- (b);    
   \draw (o) -- (c);   
   \draw (o) -- (d);     
   \draw (o) -- (e);    
   \draw (o) -- (f);    

   \draw node(aa) [double] at (0:11\mydim/13){};
   \draw node(ba) [double] at (60:11\mydim/13){};
   \draw node(ca) [double] at (120:11\mydim/13){};
   \draw node(da) [double] at (180:11\mydim/13){};
   \draw node(ea) [double] at (240:11\mydim/13){};
   \draw node(fa) [double] at (300:11\mydim/13){};   

   \draw node(ab) [single] at (0:7\mydim/13) {};
   \draw node(bb) [single] at (60:7\mydim/13) {};
   \draw node(cb) [single] at (120:7\mydim/13) {};
   \draw node(db) [single] at (180:7\mydim/13) {};
   \draw node(eb) [single] at (240:7\mydim/13) {};
   \draw node(fb) [single] at (300:7\mydim/13) {};    
   
   \draw node(ab) [single] at (0:3\mydim/13) {};
   \draw node(bb) [single] at (60:3\mydim/13) {};
   \draw node(cb) [single] at (120:3\mydim/13) {};
   \draw node(db) [single] at (180:3\mydim/13) {};
   \draw node(eb) [single] at (240:3\mydim/13) {};
   \draw node(fb) [single] at (300:3\mydim/13) {};    
      
   \draw (0:12\mydim/13) node(axi1) [above, font=\scriptsize] {$\xi$};         
   \draw (60:12\mydim/13) node(axi1) {};           
   \draw (120:12\mydim/13) node(axi1) {};          
   \draw (180:12\mydim/13) node(axi1) {};          
   \draw (240:12\mydim/13) node(axi1) {};  
   \draw (300:12\mydim/13) node(axi1) {};         

   \draw (0:9\mydim/13) node(axi2) [above, font=\scriptsize] {$2\xi$};         
   \draw (60:9\mydim/13) node(axi2) {};           
   \draw (120:9\mydim/13) node(axi2) {};           
   \draw (180:9\mydim/13) node(axi2) {};           
   \draw (240:9\mydim/13) node(axi2) {};           
   \draw (300:9\mydim/13) node(axi2) {};          

   \draw (0:5\mydim/13) node(axi3) [above, font=\scriptsize] {$2\xi$};         
   \draw (60:5\mydim/13) node(axi3) {};           
   \draw (120:5\mydim/13) node(axi3) {};           
   \draw (180:5\mydim/13) node(axi3) {};           
   \draw (240:5\mydim/13) node(axi3) {};          
   \draw (300:5\mydim/13) node(axi3) {};   
     
   \draw (0:1.5\mydim/13) node(axi3) [above, font=\scriptsize] {$y$};         
   \draw (60:2\mydim/13) node(axi3) {};           
   \draw (120:2\mydim/13) node(axi3) {};           
   \draw (180:1.5\mydim/13) node(axi3) {};           
   \draw (240:2\mydim/13) node(axi3) {};           
   \draw (300:2\mydim/13) node(axi3) {};           

\end{scope}

\end{tikzpicture}
~\vspace{0cm} \caption{\label{fi:eqstar} Equilibria on $S_{k}$ with $3k-1$ players, $4k-1$ players and $4k+r$ players (with $k/(2r+2+5k) \leq \xi \leq k/(2r+5k)$). In all the pictures $k=6$ and in the right one $y=1-5\xi$. On each star the disposition of players is the same on all the $k$ rays.}
\end{figure}

\begin{remark}\label{re:star}

For both the segment and the circle the price of anarchy has a nonmonotonic behavior in the number of players, but is always smaller than or equal $2$. 
The next example shows that this is not the case for the star, where the price of anarchy takes values that are larger than $2$ infinitely often. For the sake of simplicity, we consider the case $k=3$.

Consider the sequence of games $\mathcal{L}(n, S_{3})$ where $S_{3}$ is a  star with $3$ rays and the number of players is $n=3(2b+1)$, with $b>2$.

The worst equilibrium $\widehat{\boldsymbol{x}}$ is as follows: there are $b$ pairs of players on each ray, and $3$ players in the center. Its social cost is
\[
C(\widehat{\boldsymbol{x}})= 3 (2b+1)  \frac{1}{2(2b+1)^{2}}=\frac{3}{4b+2}.
\]
Consider now the profile $\overline{\boldsymbol{x}}$ such that one player is located at at the center and $2b+1$ players sit on each ray, except one ray that has $2b$ players. Its costs is
\[
C(\overline{\boldsymbol{x}})= 2(4b+3)  \frac{1}{2(4b+3)^2}
+ (4b+1)  \frac{1}{2(4b+1)^2}
= \frac{1}{4b+3}+\frac{1}{8b+2}.
\]
Figure~\ref{fi:star3EQOPT} shows profiles $\widehat{\boldsymbol{x}}$ and $\overline{\boldsymbol{x}}$ for the case $b=1$.

\begin{figure}[H]

\tikzstyle{leaf}=[rectangle,fill,scale=0.2]
\tikzstyle{single}=[circle,fill=blue,scale=0.5]
\tikzstyle{double}=[rectangle,fill=red,scale=0.7]
\tikzstyle{rple}=[star,fill=green,scale=0.6]
\tikzstyle{kple}=[regular polygon,regular polygon sides=3,fill=brown,scale=0.5]

\newdimen\mydim
\begin{tikzpicture}
\mydim=5cm

   \draw node(z) at (0,0) {};

\begin{scope}[scale=0.6,xshift=5cm]

   \draw node(a) [leaf] at (0:\mydim) {};
   \draw node(c) [leaf] at (120:\mydim) {};
   \draw node(e) [leaf] at (240:\mydim) {};
   \draw node(o) [rple] at (0,0) {} ;
   
   \draw (o) -- (a);     
   \draw (o) -- (c);   
   \draw (o) -- (e);    

   \draw node(aa) [double] at (0:2\mydim/3) {};
   \draw node(ca) [double] at (120:2\mydim/3) {};
   \draw node(ea) [double] at (240:2\mydim/3) {};
       
   \draw (0:\mydim/3) node(axi2) [above, font=\scriptsize] {$2/3$};         
   \draw (120:\mydim/3) node(axi2) [above, right, font=\scriptsize] {$2/3$};         
   \draw (240:\mydim/3) node(axi2) [below, right, font=\scriptsize] {$2/3$};         

   \draw (0:5\mydim/6) node(axi3) [above, font=\scriptsize] {$1/3$};         
   \draw (120:5\mydim/6) node(axi3) [above, right, font=\scriptsize] {$1/3$};         
   \draw (240:5\mydim/6) node(axi3) [below, right, font=\scriptsize] {$1/3$};         

\end{scope}

\begin{scope}[scale=0.6,xshift=15cm]
       
   \draw node(a) [leaf] at (0:\mydim) {};
   \draw node(c) [leaf] at (120:\mydim) {};
   \draw node(e) [leaf] at (240:\mydim) {};
   \draw node(o) [single] at (0,0) {} ;
   
   \draw (o) -- (a);     
   \draw (o) -- (c);   
   \draw (o) -- (e);    

   \draw node(aa) [single] at (0:2\mydim/7) {};
   \draw node(ca) [single] at (120:2\mydim/7) {};
   \draw node(ea) [single] at (240:2\mydim/5) {};

   \draw node(aa) [single] at (0:4\mydim/7) {};
   \draw node(ca) [single] at (120:4\mydim/7) {};
   \draw node(ea) [single] at (240:4\mydim/5) {};

   \draw node(aa) [single] at (0:6\mydim/7) {};
   \draw node(ca) [single] at (120:6\mydim/7) {};

   \draw (0:\mydim/7) node(axi2) [above, font=\scriptsize] {$2/7$};         
   \draw (120:\mydim/7) node(axi2) [left, font=\scriptsize] {$2/7$};         
   \draw (240:\mydim/5) node(axi2) [below, right, font=\scriptsize] {$2/5$};         

   \draw (0:3\mydim/7) node(axi2) [above, font=\scriptsize] {$2/7$};         
   \draw (120:3\mydim/7) node(axi2) [left, font=\scriptsize] {$2/7$};         
   \draw (240:3\mydim/5) node(axi2) [below, right, font=\scriptsize] {$2/5$};         

   \draw (0:5\mydim/7) node(axi2) [above, font=\scriptsize] {$2/7$};         
   \draw (120:5\mydim/7) node(axi2) [left, font=\scriptsize] {$2/7$};         

   \draw (0:13\mydim/14) node(axi2) [above, font=\scriptsize] {$1/7$};         
   \draw (120:13\mydim/14) node(axi2) [left, font=\scriptsize] {$1/7$};         
   \draw (240:9\mydim/10) node(axi2) [below, right, font=\scriptsize] {$1/5$};                 

\end{scope}

\begin{scope}[scale=0.6,xshift=22cm]
       
   \draw node [single] at (0, 1) {};
   \draw (0,1) node [right]  {\ {$1$ player}};
             
   \draw node [double] at (0, 0) {};
   \draw (0,0) node [right]  {\ {$2$ players}};
    
   \draw node [rple] at (0, -1) {};
   \draw (0, -1) node [right]  {\ {$3$ players}};

\end{scope}

\end{tikzpicture}
~\vspace{0cm} \caption{\label{fi:star3EQOPT}Left: equilibrium $\widehat{\boldsymbol{x}}$ on $S_{3}$ with $9$ players. Right: good configuration $\overline{\boldsymbol{x}}$ on $S_{3}$ with $9$ players.}
\end{figure}
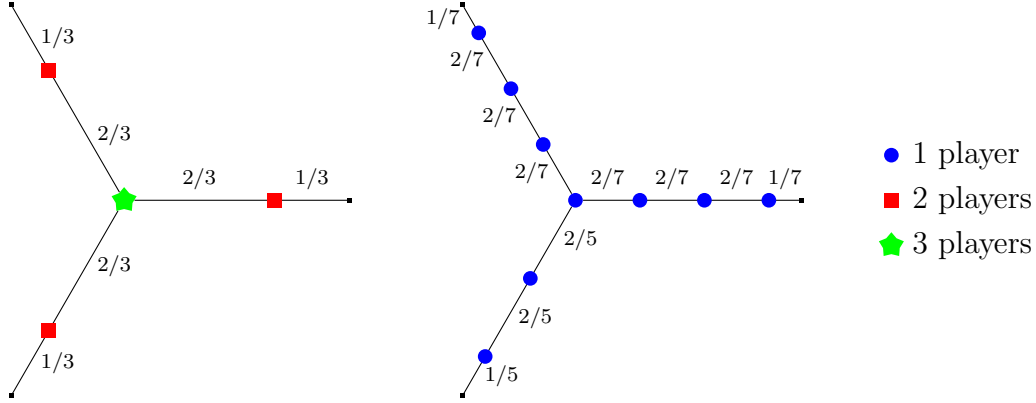

Since the cost of the optimal profile is not greater than  $C(\overline{\boldsymbol{x}})$, we have

\[
\PoA(n) \geq \frac{C(\widehat{\boldsymbol{x}})}{C(\overline{\boldsymbol{x}})}=\frac{3}{(4b+2)(\frac{1}{4b+3}+\frac{1}{8b+2})}=\frac{3}{\frac{4b+2}{4b+3}+\frac{4b+2}{8b+2}} > 2
\]
for all $b>2$. 
The last inequality follows from the fact that the denominator is equal to 
\begin{equation*}
\frac{3}{2}-\frac{1}{4b+3}+\frac{1}{8b+2}<\frac{3}{2}.
\end{equation*}
Obviously $\PoA(n) \to 2$, as $b\to\infty$.

On the other hand in location games on the star the price of anarchy takes values smaller than $2$ for infinitely many values of $n$. 
Consider the class of games $\mathcal{L}(n,S_3)$ with $n=6b+1$, $b>2$.  Figure~\ref{fi:star3rays} shows the case $b=3$. 

The worst equilibrium $\widehat{x}$ is as follow: there are $2$ players in the center, $b$ pairs of players equally spaced on $2$ rays, and $(b-1)$ pairs plus a single player equally spaced on the third ray. We have
\begin{equation*}
C(\widehat{x})=3(2b+1)\frac{1}{2(2b+1)^2}=\frac{3}{2(2b+1)}.
\end{equation*}
The optimum profile $\overline{x}$ is as follow: there is $1$ player at the center and $2b$ players equally spaced on each of the $3$ ray. We have
\begin{equation*}
C(\overline{x})=3(4b+1)\frac{1}{2(4b+1)^2}=\frac{3}{2(4b+1)}.
\end{equation*}
Therefore,
\begin{equation*}
\PoA(n)=\frac{4b+1}{2b+1} < 2.
\end{equation*}

\begin{figure}[H]

\tikzstyle{leaf}=[rectangle,fill,scale=0.2]
\tikzstyle{single}=[circle,fill=blue,scale=0.5]
\tikzstyle{double}=[rectangle,fill=red,scale=0.7]
\tikzstyle{rple}=[star,fill=green,scale=0.6]
\tikzstyle{kple}=[regular polygon,regular polygon sides=3,fill=brown,scale=0.5]

\newdimen\mydim
\begin{tikzpicture}
\mydim=5cm

   \draw node(z) at (0,0) {};

\begin{scope}[scale=0.6,xshift=5cm]
   
   \draw node(a) [leaf] at (0:\mydim) {};

   \draw node(c) [leaf] at (120:\mydim) {};

   \draw node(e) [leaf] at (240:\mydim) {};

   \draw node(o) [double] at (0,0) {} ;
   
   \draw (o) -- (a);     
   
   \draw (o) -- (c);   
     
   \draw (o) -- (e);    
    
   \draw node(aa) [double] at (0:2\mydim/7) {};
   \draw node(aa) [double] at (0:4\mydim/7) {};
   \draw node(aa) [double] at (0:6\mydim/7) {};

   \draw node(aa) [double] at (120:2\mydim/7) {};
   \draw node(aa) [double] at (120:4\mydim/7) {};
   \draw node(aa) [double] at (120:6\mydim/7) {};
   
   \draw node(aa) [single] at (240:2\mydim/7) {};
   \draw node(aa) [double] at (240:4\mydim/7) {};
   \draw node(aa) [double] at (240:6\mydim/7) {};

\end{scope}

\begin{scope}[scale=0.6,xshift=15cm]

   \draw node(a) [leaf] at (0:\mydim) {};

   \draw node(c) [leaf] at (120:\mydim) {};

   \draw node(e) [leaf] at (240:\mydim) {};

   \draw node(o) [single] at (0,0) {} ;
   
   \draw (o) -- (a);     
   
   \draw (o) -- (c);   
     
   \draw (o) -- (e);    
    
   \draw node(aa) [single] at (0:2\mydim/13) {};
   \draw node(aa) [single] at (0:4\mydim/13) {};
   \draw node(aa) [single] at (0:6\mydim/13) {};
   \draw node(aa) [single] at (0:8\mydim/13) {};
   \draw node(aa) [single] at (0:10\mydim/13) {};
   \draw node(aa) [single] at (0:12\mydim/13) {};

   \draw node(aa) [single] at (120:2\mydim/13) {};
   \draw node(aa) [single] at (120:4\mydim/13) {};
   \draw node(aa) [single] at (120:6\mydim/13) {};
   \draw node(aa) [single] at (120:8\mydim/13) {};
   \draw node(aa) [single] at (120:10\mydim/13) {};
   \draw node(aa) [single] at (120:12\mydim/13) {};
   
   \draw node(aa) [single] at (240:2\mydim/13) {};
   \draw node(aa) [single] at (240:4\mydim/13) {};
   \draw node(aa) [single] at (240:6\mydim/13) {};
   \draw node(aa) [single] at (240:8\mydim/13) {};
   \draw node(aa) [single] at (240:10\mydim/13) {};
   \draw node(aa) [single] at (240:12\mydim/13) {};
\end{scope}

\begin{scope}[scale=0.6,xshift=22cm]
       
   \draw node [single] at (0, 0) {};
   \draw (0,0) node [right]  {\ {$1$ player}};
             
   \draw node [double] at (0, -1) {};
   \draw (0,-1) node [right]  {\ {$2$ players}};
   
\end{scope}

\end{tikzpicture}
~\vspace{0cm} \caption{\label{fi:star3rays} An example of worst equilibrium (left) and social optimum (right) on $S_{3}$ with $19$ players.}
\end{figure}
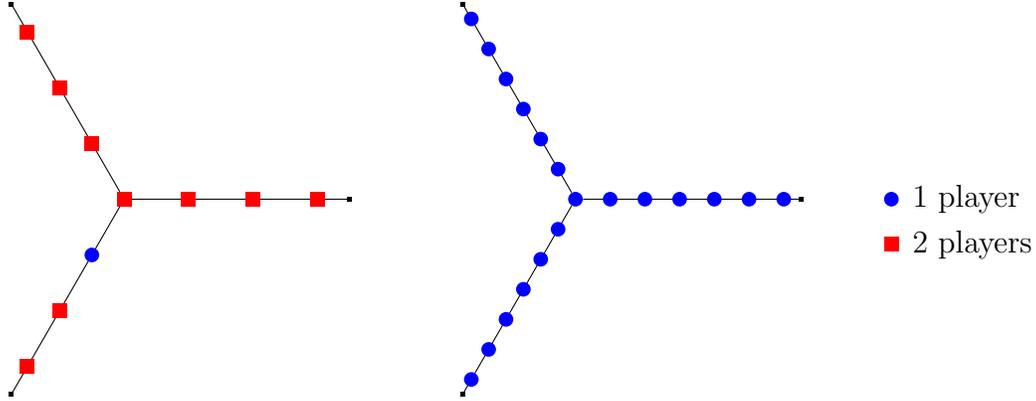

\end{remark}

\subsection{Proofs}

\subsubsection*{The circle}
The set action for the players is now a circle $\mathcal{C}$ with radius $r$. Without loss of generality we can suppose that $r=1$. We fix an arbitrary origin $0$ and identify the circle $\mathcal{C}$ with ${\mathbb{R}}/{2\pi}$ and we consider the representations in $[0,2\pi)$. We now define $3$ different profiles of locations in the game $\mathcal{L}(n, \mathcal{C})$.
\begin{enumerate}[(i)]
\item\label{it:circlebesteq}
Call $\widetilde{\boldsymbol{x}}$  the profile such that
\[
\widetilde{x}_{i} = \frac{i}{n}2\pi. 
\]

\item\label{it:circleworsteqeven}
For $n$ even call $\widehat{\boldsymbol{x}}$ the action profile such that
\[
\widehat{x}_{2i-1} = \widehat{x}_{2i} = \frac{2i}{n}2\pi, \quad i \in \{1, \dots, n/2\}. 
\]

\item\label{it:circleworsteqodd}
For $n$ odd call $\breve{\boldsymbol{x}}$ the action profile such that
\begin{align*}
&\breve{x}_{2i-1} = \breve{x}_{2i} = \frac{2i}{n}2\pi, \quad i \in \{1, \dots, \lfloor n/2\rfloor\} \\
&\breve{x}_{n} = 2\pi.
\end{align*}

\end{enumerate}

Propositions \ref{pr:existenceC} and ~\ref{pr:pricesalop} are direct consequences of the following lemma. 
Considerations about the existence of equilibria for location games on the circle were already present in \citet{EatLip:RES1975}.

\begin{lemma}\label{le:salop}
\begin{enumerate}[{\rm (a)}]

\item\label{it:le:eqbestworstcircle-profiles}
The action profiles $\widetilde{\boldsymbol{x}}$, $\widehat{\boldsymbol{x}}$, and $\breve{\boldsymbol{x}}$
are equilibria in $\mathcal{L}(n,\mathcal{C})$.

\item\label{it:le:eqbestworstcircle-check}
For all positive $n$,
\begin{align}\label{eq:beactioncircle}
\widetilde{\boldsymbol{x}} &\in \arg \min_{\boldsymbol{x}\in \mathcal{E}_{n}} C(\boldsymbol{x}),\\
\widetilde{\boldsymbol{x}} &\in \arg \min_{\boldsymbol{x}\in \mathcal{C}^{n}} C(\boldsymbol{x}), \label{eq:optactioncircle}
\end{align}
and
\begin{equation}\label{eq:bevaluecircle}
C(\widetilde{\boldsymbol{x}}) = \frac{\pi^{2}}{n}.
\end{equation}

\item\label{it:le:eqbestworstcircle-hateven}
For $n$ even
\begin{equation}\label{eq:weactionnevencircle}
\widehat{\boldsymbol{x}} \in \arg \max_{\boldsymbol{x}\in \mathcal{E}_{n}} C(\boldsymbol{x}) 
\end{equation}
and
\begin{equation}\label{eq:wevaluenevencircle}
C(\widehat{\boldsymbol{x}}) =  \frac{2\pi^{2}}{n}.
\end{equation}

\item\label{it:le:eqbestworstcircle-hatodd}
For $n$ odd
\begin{equation}\label{eq:weactionnoddcircle}
\breve{\boldsymbol{x}} \in \arg \max_{\boldsymbol{x}\in \mathcal{E}_{n}} C(\boldsymbol{x}) 
\end{equation}
and
\begin{equation}\label{eq:wevaluenoddcircle}
C(\breve{\boldsymbol{x}}) = \frac{2\pi^{2}}{n+1}.
\end{equation}
\end{enumerate}

\end{lemma}

We define the simplex
\begin{equation}\label{eq:simplex}
\simplex_{n}(a):=\left\{(y_{1},\dots,y_{n})\in\mathbb{R}^{n}_{+}:\sum_{i=1}^{n}y_{i}=a\right\}.
\end{equation}

\begin{proof}[Proof of Lemma~\ref{le:salop}]
\noindent\ref{it:le:eqbestworstcircle-profiles}\ref{it:circlebesteq}
In the profile $\widetilde{\boldsymbol{x}}$ every player's payoff is $2\pi/n$. 
The length of any interval between two consecutive players is $2 \pi/n$ so a unilateral deviation to such an interval is not profitable. 
A unilateral deviation to a location already occupied by another player would produce a payoff equal to either $3 \pi/(2n)$ or $\pi/n$, depending on whether this player is or is not a neighbor. In both cases such a deviation is not profitable, hence, $\widetilde{\boldsymbol{x}}$ is a Nash equilibrium.

\ref{it:circleworsteqeven} In the profile $\widehat{\boldsymbol{x}}$  every player's payoff is $2\pi/n$. 
The length of an interval between two consecutive chosen locations is $4\pi/n$, so a unilateral deviation in such an interval is not strictly profitable. A unilateral deviation to a location already chosen by another pair of players would produce a payoff equall to $4 \pi/(3n)$, and would not be profitable. Hence, $\widehat{\boldsymbol{x}}$ is a Nash equilibrium.

\ref{it:circleworsteqodd} In the profile $\breve{\boldsymbol{x}}$,  every player's payoff is larger than $2 \pi/(n+1)$. The length of an interval between two consecutive chosen locations is $4 \pi/(n+1)$ so a unilateral deviation in such an interval is not strictly profitable. A unilateral deviation to a location already chosen by a single player would produce a payoff of $2 \pi/(n+1)$, and is not profitable. A unilateral deviation to a position already chosen by a pair of players would produce a payoff equal to either $2 \pi/(n+1)$  or $4 \pi/(3n+3)$ depending on whether the deviator is or is not the single player that deviates to one of her neighbor's location. In both cases the deviation would not be profitable and we can conclude that $\widehat{\boldsymbol{x}}$ is a Nash equilibrium.

\noindent\ref{it:le:eqbestworstcircle-check} 
We remind the reader that $\boldsymbol{\lambda}(\boldsymbol{x})$ is the vector of all $\lambda([a,b])$ such that $[a,b]\in H(\boldsymbol{x})$, where $H(\boldsymbol{x})$ is the class of all $\boldsymbol{x}$-half intervals in $S$. With these notations, we have that $\boldsymbol{\lambda}(\widetilde{\boldsymbol{x}})=(\pi/n,\dots,\pi/n) \in \simplex_{2n}(2\pi)$. 
This vector is majorized by any vector in $\simplex_{2n}(2\pi)$. 
Since the mapping $z \mapsto z^{2}/2$ is convex, we can apply Lemma \ref{le:sum_convex} and conclude by Schur convexity that $\widetilde{\boldsymbol{x}} \in \arg \min_{\boldsymbol{x}\in \mathcal{C}^{n}} C(\boldsymbol{x})$. 
Moreover, $\widetilde{\boldsymbol{x}}$ is an equilibrium, hence, $\widetilde{\boldsymbol{x}} \in \arg \min_{\boldsymbol{x}\in \mathcal{E}_{n}} C(\boldsymbol{x})$.
The cost $C(\widetilde{\boldsymbol{x}})$ can be written as
\[
C(\widetilde{\boldsymbol{x}})= \sum_{[a,b]\in H(\widetilde{\boldsymbol{x}})} \frac{\lambda([a,b])^2}{2}= \sum_{i=1}^{2n} \frac{\pi^2}{2n^2} =\frac{\pi^2}{n}.
\]

\noindent\ref{it:le:eqbestworstcircle-hateven}
The vector $\boldsymbol{\lambda}(\widehat{\boldsymbol{x}})=(2\pi/n,\dots,2\pi/n,0,\dots,0) \in \simplex_{2n}(2\pi)$, is such that its $n$ first components are $2\pi/n$ and the remaining $n$ are $0$. 
It follows from Lemma~\ref{le:equilibriumhalfinterval} that at equilibrium, no half intervals can have a length longer than $2\pi/n$. Any vector in $ \simplex_{2n}(2\pi)$ whose components are bounded by  $2\pi/n$ is majorized by $\boldsymbol{\lambda}(\widehat{\boldsymbol{x}})$. 
We conclude by Schur convexity that 
$\widehat{\boldsymbol{x}} \in \arg \max_{\boldsymbol{x}\in \mathcal{E}_{n}} C(\boldsymbol{x})$. Moreover,
\[C(\widehat{\boldsymbol{x}})= \sum_{[a,b]\in H(\widehat{\boldsymbol{x}})} \frac{\lambda([a,b])^2}{2}= \sum_{i=1}^{n} \frac{4\pi^2}{2n^2} =\frac{2\pi^2}{n}.
\]

\noindent\ref{it:le:eqbestworstcircle-hatodd}
The vector $\boldsymbol{\lambda}(\breve{\boldsymbol{x}})=(2\pi/(n+1),\dots,2\pi/(n+1),0,\dots,0) \in \simplex_{2n}(2\pi)$, has the first $n+1$ components equal to $2\pi/(n+1)$ and the remaining $n-1$ equal to $0$. Since $n$ is odd, in equilibrium there are at most $(n-1)/2$ pairs of players who share the same location, and therefore at most $n-1$ half intervals with length $0$. Therefore, for any equilibrium $\boldsymbol{x}^{*}$, the vector $\boldsymbol{\lambda}(\boldsymbol{x}^{*})$ is majorized by $\boldsymbol{\lambda}(\breve{\boldsymbol{x}})$ and we can conclude that $\breve{\boldsymbol{x}} \in \arg \max_{\boldsymbol{x}\in \mathcal{E}_{n}} C(\boldsymbol{x})$.
Moreover
\begin{equation*}
C(\breve{\boldsymbol{x}})= \sum_{[a,b]\in H(\breve{\boldsymbol{x}})} \frac{\lambda([a,b])^2}{2}= \sum_{i=1}^{n+1} \frac{4\pi^2}{2(n+1)^2} =\frac{2\pi^2}{n+1}. \qedhere
\end{equation*}
\end{proof}

\subsubsection*{The segment}

We refer to \citet{EatLip:RES1975} or \citet{Pal:mimeo2011} for a proof of Proposition~\ref{pr:zeroone}. Proposition~\ref{pr:ipoaipossegment} is a direct consequence of Lemmata~\ref{le:optsegment} and \ref{le:eqbestworstsegment} below.

\begin{lemma}\label{le:optsegment}
If for all $i \in N$ we define
\[
\overline{x}_{i} = \frac{2i-1}{2n},
\]
then 
\begin{equation}\label{eq:optaction}
\overline{\boldsymbol{x}} \in \arg \min_{\boldsymbol{x}\in [0,1]^{n}} C(\boldsymbol{x})
\end{equation}
and
\begin{equation}\label{eq:optvalue} 
C(\overline{\boldsymbol{x}}) = \frac{1}{4n}.
\end{equation}
\end{lemma}

\begin{proof}
$\boldsymbol{\lambda}(\overline{\boldsymbol{x}})=(1/(2n),\dots,1/(2n)) \in \simplex_{2n}(1)$ is majorized by any vector in $\simplex_{2n}(1)$. Hence $\overline{\boldsymbol{x}}$ minimizes $C(\cdot)$ and 
\begin{equation*}
C(\overline{\boldsymbol{x}})= \sum_{[a,b]\in H(\bar{\boldsymbol{x}})} \frac{\lambda([a,b])^2}{2} = \sum_{i=1}^{2n} \frac{1}{8n^2}=\frac{1}{4n}. \qedhere
\end{equation*}
\end{proof}

\begin{lemma}\label{le:eqbestworstsegment}
Consider the game $\mathcal{L}(n,[0,1])$ and the action profiles $\widetilde{\boldsymbol{x}}$, $\widehat{\boldsymbol{x}}$, and $\widehat{\boldsymbol{x}}^{\ell}$, for  $\ell\in\{2,\dots, (n-3)/2 \}$, defined as follows: 

For $n \geq 5$ 
\begin{align*}
& \widetilde{x}_{1}=\widetilde{x}_{2}=\frac{1}{2n-4}, \\
& \widetilde{x}_{i}=\frac{2i-3}{2n-4}, \qquad i \in \{3, \dots, n-2\}, \\
& \widetilde{x}_{n-1}=\widetilde{x}_{n}= \frac{2n-5}{2n-4}.
\end{align*}

For $n$ even
\[
\widehat{x}_{2i-1}=\widehat{x}_{2i}=\frac{2i-1}{n}, \qquad i \in \{1, \dots, n/2\}.
\]

For $n$ odd and $\ell\in\{2,\dots, (n-3)/2 \}$
\begin{align*}
& \widehat{x}^{\ell}_{2i-1}=\widehat{x}^{\ell}_{2i}=\frac{2i-1}{n+1}, \qquad i \in \{1, \dots, \ell\}, \\
& \widehat{x}^{\ell}_{2\ell+1}=\frac{2\ell+1}{n+1}, \\
& \widehat{x}^{\ell}_{2i}=\widehat{x}^{\ell}_{2i+1}=\frac{2i+1}{n+1}, \qquad i \in \left\{\ell+1, \dots,  \frac{n-1}{2} \right\}.
\end{align*}

\begin{enumerate}[{\rm (a)}]
\item\label{it:le:eqbestworstsegment-profiles}
The action profiles $\widetilde{\boldsymbol{x}}$, $\widehat{\boldsymbol{x}}$, and $\widehat{\boldsymbol{x}}^{\ell}$
are equilibria in $\mathcal{L}(n,[0,1])$.

\item\label{it:le:eqbestworstsegment-check}
For all $n \ge 4$
\begin{equation}\label{eq:beaction}
\widetilde{\boldsymbol{x}} = \arg \min_{x\in \mathcal{E}_{n}} C(\boldsymbol{x}) 
\end{equation}
and
\begin{equation}\label{eq:bevalue}
C(\widetilde{\boldsymbol{x}}) = \frac{1}{4(n-2)}.
\end{equation}

\item\label{it:le:eqbestworstsegment-hateven}
For $n$ even
\begin{equation}\label{eq:weactionneven}
\widehat{\boldsymbol{x}} = \arg \max_{x\in \mathcal{E}_{n}} C(\boldsymbol{x}) 
\end{equation}
and
\begin{equation}\label{eq:wevalueneven}
C(\widehat{\boldsymbol{x}}) =  \frac{1}{2n}.
\end{equation}

\item\label{it:le:eqbestworstsegment-hatodd}
For $n > 3$ odd
\begin{equation}\label{eq:weactionnodd}
\widehat{\boldsymbol{x}}^{\ell} = \arg \max_{x\in \mathcal{E}_{n}} C(\boldsymbol{x}) 
\end{equation}
and
\begin{equation}\label{eq:wevaluenodd}
C(\widehat{\boldsymbol{x}}^{\ell}) = \frac{n}{2(n+1)}.
\end{equation}

\end{enumerate}
\end{lemma}

\begin{proof}
\noindent\ref{it:le:eqbestworstsegment-profiles}
In the profile $\widetilde{\boldsymbol{x}}$ every player's payoff is $1/(n-2)$ and there is no interval between two players (or between a player and a leaf) with length larger than $1/(n-2)$, so no profitable deviation in the interior of an interval is possible. 
A deviation to a location with a single player (resp. a pair of players) would induce a payoff of at most $3/(4n-8)$ (resp. $1/(2n-4)$) and is therefore not profitable.

In the profile $\widehat{\boldsymbol{x}}$ every player's payoff is $1/n$. A deviation in an interval between two consecutive players (or between a player and a leaf) would induce a payoff equal to $1/n$ or less, and a deviation to an occupied location would induce a payoff of $2/(3n)$. Both deviations are unprofitable.

In the profile $\widehat{\boldsymbol{x}}^{\ell}$ every player has a payoff equal to either $1/(n+1)$ or $2/(n+1)$. A deviation to an interval between two players (or between a player and a leaf) would induce a payoff smaller than $1/(n+1)$. A deviation to the location occupied by a single player would induce a payoff of $1/(n+1)$ and a deviation to a location occupied by a pair of players would induce a payoff of $1/(n+1)$ or $2/(3n+3)$, irrespective of whether the deviator is or is not the single player deviating to one of its neighbor location. In both cases, the deviation is not profitable.

\noindent\ref{it:le:eqbestworstsegment-check}
In this case
\[
\boldsymbol{\lambda}(\widetilde{\boldsymbol{x}})=\left(\frac{1}{2n-4},\dots,\frac{1}{2n-4},0,0,0,0\right) \in \simplex_{2n}(1),
\] 
i.e., the $2n-4$ first components are equal and positive and the last $4$ are $0$.
Since at equilibrium we necessarily have  $x_{1}=x_{2}$ and $x_{n-1}=x_{n}$, for any equilibrium $\boldsymbol{x}^{*}$, the vector $\boldsymbol{\lambda}(\boldsymbol{x}^{*})$ must contain $4$ null components. The vector $\boldsymbol{\lambda}(\widetilde{\boldsymbol{x}})$ is majorized by any  vector in $\simplex_{2n}(1)$ with $4$ null components, therefore,
$\widetilde{\boldsymbol{x}} \in \arg \min_{\boldsymbol{x}\in \mathcal{E}_{n}} C(\boldsymbol{x})$. 
Moreover,
\[\min_{x\in \mathcal{E}_{n}} C(\boldsymbol{x})  = \sum_{[a,b]\in H(\widetilde{\boldsymbol{x}})} \frac{\lambda([a,b])^2}{2} = \sum_{i=1}^{2n-4} \frac{1}{2(2n-4)^2}= \frac{1}{4(n-2)}.
\]

\noindent\ref{it:le:eqbestworstsegment-hateven}
In this case
\[
\boldsymbol{\lambda}(\widehat{\boldsymbol{x}})=\left(\frac{1}{n},\dots,\frac{1}{n},0, \dots ,0\right) \in \simplex_{2n}(1),
\] 
i.e., the first $n$  components are  equal to $1/n$  and the last $n$ are $0$.
According to Lemma~\ref{le:equilibriumhalfinterval}, in equilibrium the length of any half interval is at most $1/n$. Any  vector in $\simplex_{2n}(1)$ with components smaller than $1/n$ is majorized by $\boldsymbol{\lambda}(\widehat{\boldsymbol{x}})$, hence, $\widehat{\boldsymbol{x}} \in \arg \max_{\boldsymbol{x}\in \mathcal{E}_{n}} C(\boldsymbol{x})$. Moreover, 
\[
\max_{x\in \mathcal{E}_{n}} C(\boldsymbol{x})  = \sum_{[a,b]\in H(\widehat{\boldsymbol{x}})} \frac{\lambda([a,b])^2}{2} = \sum_{i=1}^{n} \frac{1}{2(n)^2}= \frac{1}{2n}.
\]

\noindent\ref{it:le:eqbestworstsegment-hatodd}
The argument is similar to the one used to prove \ref{it:le:eqbestworstsegment-hateven}.
\end{proof}

The optimum and equilibrium profiles used in Lemmata~\ref{le:optsegment} and \ref{le:eqbestworstsegment} are depicted in Figure~\ref{fi:opteqsegment}.

\subsubsection*{The star}

For $j \in \{1, \dots, k\}$ define
$
N^{j} = \left\{i \in N: x_{i} \in  e_{v_{0}v_{j}}\setminus \{v_{0}\}\right\}
$
and call $h(j)$ the cardinality of $N^{j}$. Order the players $i_{j,1} \prec \dots \prec i_{j,h(j)} \in N^{j}$ in terms of the distance of their actions $x_{i_{j,1}}, \dots, x_{i_{j,h(j)}}$ from $v_{0}$, from the smallest to the largest (solve the ties arbitrarily).

For $j\in\{1,\dots,k\}$, define 
\[
p_{j}(\boldsymbol{x})=
\begin{cases}
d(v_{0},x_{i_{j,1}})/2 & \text{if there is at least one player on $e_{v_{0}v_{j}}$},\\
1 & \text{otherwise}. 
\end{cases}
\]
To simplify the notation we call  
\[
g_{i_{j,\ell}} = \rho_{i_{j,\ell}}(\boldsymbol{x}^{*}).
\]
Therefore, if the profile $\boldsymbol{x}$ is such that there is a player in $v_{0}$, then
\begin{enumerate}[{\rm (a)}]
\item
the payoff of the player in $v_{0}$ is 
\[
\frac{\sum_{j=1}^{k} p_{j}(\boldsymbol{x})}{\card\{\ell: x_{\ell} = v_{0}\}},
\]

\item
for $j$ in $\{1, \dots ,k \}$ 
we have
\begin{align*}
g_{i_{j,1}} &= \frac{d(v_{0}, x_{i_{j,1}})+ d(x_{i_{j,1}},x_{i_{j,2}})}{2\card  \{m: x_m = x_{i_{j,1}}\}},\\
g_{i_{j,\ell}} &= \frac{ d(x_{i_{j,\ell-1}}, x_{i_{j,\ell}}) +  d(x_{i_{j,\ell}}, x_{i_{j,\ell+1}})}{2\card  \{m: x_m = x_{i_{j,\ell}}\}}\quad\text{for }\ell \in \{ 2, \dots , h(j)-1 \},\\ 
g_{i_{j,h(j)}} &= \frac{2d( x_{i_{j,h(j)}},v_{j}) + d(x_{i_{j,h(j)}},x_{i_{j,h(j)-1}})
}{2\card  \{m: x_m = x_{i_{j,h(j)}}\}}.
\end{align*}

\end{enumerate}

Let $\bar{j}= \arg \min_{j}$  $p_{j}$ and assume that in the profile $\boldsymbol{x}$ there is no player in $v_{0}$. Then, considering that the closest player to the vertex attracts also some consumers on other edges, we have that for $j$ in $\{1, \dots ,k \}$ 
\begin{equation*}
g_{i_{j,1}} =
\begin{cases}
\displaystyle{\frac{d(x_{i_{j,1}},x_{i_{j},2}) + d(x_{i_{j,1}},x_{i_{\bar{j},1}})}{2\card  \{m: x_m = x_{i_{j,1}}\}}} &\text{for }j \ne \bar{j}, \\
\displaystyle{\frac{2p_{\bar{j}}+ \sum_{j \ne \bar{j}} (p_{j}-p_{\bar{j}})+d(x_{i_{\bar{j},1}},x_{i_{\bar{j},2}})}{2\card  \{m: x_m = x_{i_{\bar{j},1}}\}}} &\text{for }j = \bar{j}. \\
\end{cases}
\end{equation*}
with $x_{i_{j,2}}=v_{j}$ if $x_{i_{j,1}}$ is the only location on the edge $j$.

\begin{lemma}\label{le:basic}
Let $\boldsymbol{x}^{*}$ be a Nash equilibrium of the location game $\mathcal{L}(n, S_{k})$. Then the following hold:
\begin{enumerate}[{\rm (a)}]
\item\label{it:le:basic-a}
There exists $i^{*} \in \{1,...,N\}$ such that $x^{*}_{i^{*}}=v_{0}$.

\item\label{it:le:basic-b}
For every $x \in S \setminus \{v_{0}\}$, $\card\{i:x^{*}_{i} = x\} \le 2$.

\item\label{it:le:basic-c}
$\card\{i:x^{*}_{i} = v_{0}\} \le k$.

\item\label{it:le:basic-d}
If for some  $i \in N$ we have $x^{*}_{i} \in S\setminus\{v_{0}\}$, then for each $j \in \{1, \dots, k\}$ there exist at least two players $i_{j,h(j)}, i_{j,h(j)-1} \in N^{j}$ such that $x^{*}_{i_{j,h(j)}} = x^{*}_{i_{j,h(j)-1}}$.
\end{enumerate}
\end{lemma}

\begin{proof}
\noindent\ref{it:le:basic-a}
Consider a star $S_{k}$ with $k\ge 3$ and assume \emph{ad absurdum} that no player is in $v_{0}$. Consider the  player $i$ such that $x^{*}_{i}$ is the closest position to $v_{0}$. For $\varepsilon<d(x^{*}_{i},v_{0})$, if player $i$ moves of $\varepsilon$ towards the center, then she loses $\varepsilon/2$ on the edge where she is and gains $\varepsilon/2$ on every other edge. Therefore moving towards the center is profitable.

\noindent\ref{it:le:basic-b}
and
\noindent\ref{it:le:basic-c}
These are particular cases of Lemma~\ref{le:playersdegree}.

\noindent\ref{it:le:basic-d} Suppose that there exists $i \in N$ such that $x^{*}_{i} \in S\setminus \{v_{0}\}$. Then player $i$ cannot be alone on her edge: if she were,  she would have a profitable deviation by moving towards the center. If one edge were empty, then any of the players could profitably deviate by moving to the empty edge, close enough to $v_{0}$.
\end{proof}

\begin{lemma}\label{le:basic2}
Let $\boldsymbol{x}^{*}$ be a Nash equilibrium of $\mathcal{L}(n, S_{k})$ and let $y \in e_{v_{0},v_{j}}\setminus \{v_{0}\},$ be such that 
\begin{equation*}
\card\{\ell:x^{*}_{\ell} = y\} = 2,
\end{equation*}
and call $i_{\ell}$ and $i_{\ell+1}$ the two players in $y$. We have
\begin{enumerate}[{\rm (a)}]
\item\label{it:le:basic2-a}
if $h(j)>\ell+1$, then
\[
g_{i_{j,\ell}}  = g_{i_{j,\ell+1}}
=d(x^{*}_{i_{j,\ell}},x^{*}_{i_{j,\ell-1}})=d(x^{*}_{i_{j,\ell+1}},x^{*}_{i_{j,\ell+2}})=:\xi(y).
\]

\item\label{it:le:basic2-b}
if $h(j)=\ell+1$, then
\[
g_{i_{j,\ell}}  = g_{i_{j,\ell+1}}
=d(x^{*}_{i_{j,\ell}},x^{*}_{i_{j,\ell-1}})=\frac{1}{2}d(x^{*}_{i_{j,\ell+1}},v_{j})=:\xi(y). 
\]

\item\label{it:le:basic2-c}
The value $\xi(y)$ does not depend on $y$ (hence we simply denote it $\xi$). 
\end{enumerate}
\end{lemma}

\begin{proof}
This follows directly from Corollary~\ref{co:equaldelta}.
\end{proof}

\begin{lemma}\label{le:carddegree}
Let $\boldsymbol{x}^{*}$ be a Nash equilibrium of $\mathcal{L}(n, S_{k})$.
If $\card\{\ell:x^{*}_{\ell} = v_{0}\} = \degree(v_{0}) = k$,
then 
\[
p_{j}(\boldsymbol{x}^{*})=\xi.
\]
\end{lemma}

\begin{proof}
This also follows from  Corollary~\ref{co:equaldelta}.
\end{proof}

\begin{proof}[Proof of Proposition~\ref{pr:barStark}]
\noindent\ref{it:pr:barStark-a}
First we prove that the profile $\boldsymbol{x}^{*}$ such that $x^{*}_{i} = v_{0}$ for all $i \in N$ is indeed an equilibrium. If any player $i$ deviates, then she will obtain a payoff that is strictly less than $1$, whereas by not deviating she obtains $k/n \ge 1$.

Now we turn to prove uniqueness. Assume by contradiction that there exists an equilibrium such that for some $i \in N$ we have $x_{i} \in e_{v_{0}v_{j}}$. Then, by Lemma~\ref{le:basic}\ref{it:le:basic-d}, each edge $e_{v_{0}v_{j}}$ has been chosen by at least two players. This implies that $n \ge 2k$, which is impossible, since $n \le k$.

\medskip
\noindent\ref{it:pr:barStark-b}
Assume by contradiction that an equilibrium $\boldsymbol{x}^{*}$ exists. If $x^{*}_{i} = v_{0}$ for all $i \in N$, then each player gains $k/n < 1$, so a profitable deviation is possible. 

We consider now the case where for some $i \in N$ we have $x_{i} \neq v_{0}$. 

First we consider the case $k < n < 2k$. 
If for some $i \in N$ we have $x^{*}_{i} \neq v_{0}$, then, by Lemma~\ref{le:basic}\ref{it:le:basic-d}, each ray contains $2$ players in the same position, so $2k$ players choose an action different from $v_{0}$, which is impossible since $n < 2k$.

If $n=2k$, and for some $i \in N$ we have $x^{*}_{i} \neq v_{0}$, then, like in the previous case, by  Lemma~\ref{le:basic}\ref{it:le:basic-d}, $x^{*}_{i} \neq v_{0}$ for all $i \in N$ and, by  Lemma~\ref{le:basic2}, $d(x^{*}_{i}, v_{0})$ is the same for all $i\in N$. If $d(x^{*}_{i}, v_{0}) < 1/2$, then, for $\varepsilon$ small enough, one player profits by choosing a position at a distance $d(x^{*}_{i}, v_{0}) + \varepsilon$ from $v_{0}$. 
If $d(x^{*}_{i}, v_{0}) \ge 1/2$, then one player profits by deviating in $v_{0}$.

Assume now $2k < n < 3k-1$. The profile where all the players choose $v_{0}$ is not an equilibrium. By Lemma~\ref{le:basic}\ref{it:le:basic-d}, for each $j \in \{1, \dots, k\}$, there exist at least two players $i_{j,h(j)-1}, i_{j,h(j)} \in N^{j}$ such that $d(x^{*}_{i_{j,h(j)-1}}, v_{j}) = d(x^{*}_{i_{j,h(j)}}, v_{j}) = \xi$. The equilibrium action of the remaining $n-2k$ players must be $v_{0}$. If this were not the case, then, for some $j \in \{1, \dots, k\}$, there would be three players on the edge $j$ and
\[
1 \leq d(v_{0},x^{*}_{j,1}) + d(x^{*}_{j,1},x^{*}_{j,h(j)})+ d(x^{*}_{j,h(j)},v_{j})= d(v_{0},x^{*}_{j,1}) + 3 \xi.
\]
On the other hand, there are at most $k-2$ remaining players, so there is an other edge with only two players. This implies $1 = 3 \xi$, which is a contradiction. Therefore the remaining $k-2$ players must be in $v_{0}$.
The payoff of every player $i$ such that $x^{*}_{i} = v_{0}$ equals $k/(3n-6k)$.
The payoff of every other player is
$1/3$ and therefore any of them would have an incentive to deviate to $v_{0}$, gaining  $k/(3n-6k+3)$ which is larger than $1/3$ when $n < 3k-1$.

\noindent\ref{it:pr:barStark-c} It is easy to prove that a profile $\boldsymbol{x}^{*}$ where on each edge two players sit at a distance $2/3$ from the origin and the remaining players sit at $v_{0}$ is an equilibrium. We now show uniqueness. Indeed we know from Lemma~\ref{le:basic}\ref{it:le:basic-c} that a profile where all players choose $v_{0}$ is not an equilibrium; moreover Lemma~\ref{le:basic}\ref{it:le:basic-d} implies that each edge has at least two players. Using the same argument that we used in the proof of \ref{it:pr:barStark-b}, we can show that no edge can have three players if another edge has only two. By Lemma~\ref{le:basic}\ref{it:le:basic-a}, at least one player chooses $v_{0}$, therefore it is not possible to have three players on each edge, if $n\in\{3k-1,3k\}$. Hence all remaining players are in $v_{0}$.

\noindent\ref{it:pr:barStark-d} We now assume $n \ge 3k+1$. Let $ n=m k + r$ be the Euclidean division of $n$ by $k$. We will construct an equilibrium $\boldsymbol{x}^{*}$ with $m$ players on each edge and $r$ players in the center, like in Figure~\ref{fi:eqstar}. Let $(2m-3)\xi + y=1$. This profile is indeed an equilibrium if and only if the following conditions are satisfied:
\begin{enumerate}[(i)]
\item
None of the $r$ players in $v_{0}$ has an incentive to deviate to an interval of length $2 \xi$, that is,
for all $i \in N$ such that $x^{*}_{i}=v_{0}$, we have $\rho_{i}(\boldsymbol{x}^{*}) = ky/(2r) \ge \xi$, which implies $y \ge 2r\xi/k$.

\item
No player has an incentive to deviate to $v_{0}$. Given that $\rho_{i}(\boldsymbol{x}^{*}) \ge \xi$ for all $i \in N$ such that $x^{*}_{i} \neq v_{0}$, we have $ky/(2r+2) \le \xi$, which implies $y \le 2(r+1)\xi/k$. 

\item
No player has an incentive to deviate to an interval of length $y$, that is $y \le 2\xi$. 

\item 
No player has an incentive to deviate to a location with another single player.
If she did, her payoff would be either $\xi$ or 
\[
\frac{\xi}{2}+\frac{y}{4} \leq \xi.
\] 
\end{enumerate}

Then, for any $\xi$ such that
\[
\frac{k}{2(r+1)+2km-3k} \le \xi \le \frac{k}{2r+2km-3k} 
\]
the profile $\boldsymbol{x}^{*}$ is an equilibrium. Hence the game has an infinite number of pure Nash equilibria.
\end{proof}

\subsection*{Acknowledgments}
The authors thank two referees, the Associate Editor, and the Area Editor for their useful suggestions.

\bibliographystyle{artbibst}
\bibliography{bibhotelling}

\end{document}